\RequirePackage[loading]{tracefnt}
\documentclass[10pt,twocolumn,twoside] {IEEEtran}

\usepackage[normalem]{ulem}

\newcommand{\revisedtwo}[1]{%
\ifx\highlightrevisionstwo\undefined{#1}%
\else\textcolor{red}{#1}%
\fi}

\newcommand{\erasedtwo}[1]{%
\ifx\showerasedtwo\undefined%
\else{\sout{#1}}%
\fi}

\newcommand{\reviewercommentnumtwo}[1]{%
\ifx\showreviewercommentnumtwo\undefined%
\else{[\bf{#1}] }%
\fi}

\newcommand{\revised}[1]{%
\ifx\highlightrevisions\undefined{#1}%
\else\textcolor{red}{#1}%
\fi}

\newcommand{\reviewercommentnum}[1]{%
\ifx\showreviewercommentnum\undefined%
\else{[\bf{#1}] }%
\fi}

\newcommand{\erased}[1]{%
\ifx\showerased\undefined%
\else{\sout{#1}}%
\fi}
\usepackage[T1]{fontenc}

\usepackage{relsize}

\usepackage{url}
\usepackage{bbm}
\usepackage{cite}
\usepackage{graphicx}
\usepackage{multirow}
\usepackage{booktabs}
\usepackage{mathtools}
\usepackage{algpseudocode}
\usepackage{algorithm}
\usepackage{amsthm,amsfonts,amsmath,amssymb}
\usepackage{graphicx,url}
\usepackage{helvet}
\usepackage{tabularx}
\usepackage{color}
\usepackage{float}

\usepackage{array}
 
 \newcommand{\noise}{\mathcal{Z}}
 \newcommand{\Reg}{\Lambda}
 \newcommand\numberthis{\addtocounter{equation}{1}\tag{\theequation}}
\newcommand{\nflops}{\texttt{\#flops}}

\newcommand{\fro}[1]{\| #1 \|_\text{fro}}
\newcommand{\nuc}[1]{\| #1 \|_\text{nuc}}
\newcommand{\msv}[1]{\| #1 \|_\text{msv}}
\newcommand{\normi}[2]{ \| #1 \|_{(#2)}}
\newcommand{\dnormi}[2]{ \| #1 \|_{(#2)}^*}

\newcommand{\rank}{r}
\newcommand{\block}{b}
\newcommand{\complex}{\mathbb{C}}
\newcommand{\inner}[2]{ \langle  #1 \, , #2 \rangle }
\newcommand{\transpose}[1] { #1^\top }

\newcommand{\proj}{ \mathcal{P} }

\newcommand{\svt}{ \textsc{SVT} }

\newcommand{\bsvt}{ \textsc{BlockSVT} }

\newcommand{\atom}{  E }

\newcommand{\dual}{ Q }

\newcommand{\epsi}[1]{  \dnormi{ \epsilon_{#1} }  {#1} }

\newcommand{\levels}{  L }

\newcommand{\obj}[1]{  \sum_{i=1}^{\levels}  \lambda_i \normi{ #1 }{i} }

\newcommand{\sumb}[1]{ \sum_{\block \in \part_i} #1}
\newcommand{\sumi}[1]{ \sum_{i=1}^{\levels} #1 }

\newcommand{\seti}[1]{ \{ #1 \}_{i=1}^\levels }

\newcommand{ \part} {P}

\newcommand{\Ub}{U_\block}
\newcommand{\Sb}{S_\block}
\newcommand{\VTb}{V^\top_\block}
\newcommand{\Vb}{V_\block}
\newcommand{\Rb}{R_\block}

\theoremstyle{definition}

\newtheorem{theorem}{Theorem}[section]

\newtheorem{prop}[theorem]{Proposition}

\newtheorem{lemma}[theorem]{Lemma}

\begin{document}

\title{Beyond Low Rank + Sparse: \\ Multi-scale Low Rank Matrix Decomposition}
\author{
    \IEEEauthorblockN{Frank~Ong, and~Michael~Lustig}\\
    	\thanks{F. Ong and M. Lustig are with the Department of Electrical Engineering and Computer Sciences, University of California, Berkeley, CA 94709 USA (e-mail: frankong@berkeley.edu and mlustig@eecs.berkeley.edu). This work is supported by NIH grants R01EB019241, R01EB009690, and
P41RR09784, Sloan research fellowship, Okawa research grant and NSF GRFP.}
}

\bibliographystyle{IEEEtran.bst}

\maketitle


\begin{abstract}

\revised{\erased{Low rank methods allow us to capture globally correlated components within matrices.}} \revisedtwo{\erasedtwo{The recent low rank + sparse matrix decomposition \revised{\erased{further}}enables us to extract sparse entries along with globally correlated low rank components. In this paper, we present a natural generalization}} We present a natural generalization of the recent low rank + sparse matrix decomposition and consider the decomposition of matrices into components of multiple scales. Such decomposition is well motivated in practice as data matrices often exhibit local correlations in multiple scales. Concretely, we propose a multi-scale low rank modeling that represents a data matrix as a sum of block-wise low rank matrices with increasing scales of block sizes. We then consider the inverse problem of decomposing the data matrix into its multi-scale low rank components and approach the problem via a convex formulation. Theoretically, we show that under various incoherence conditions, the convex program recovers the multi-scale low rank components \revised{either exactly or approximately}. Practically, we provide guidance on selecting the regularization parameters and incorporate cycle spinning to reduce blocking artifacts.  Experimentally, we show that the multi-scale low rank decomposition provides a more intuitive decomposition than conventional low rank methods and demonstrate its effectiveness in four applications, including illumination normalization for face images, motion separation for surveillance videos, multi-scale modeling of the dynamic contrast enhanced magnetic resonance imaging and collaborative filtering exploiting age information.

\end{abstract}


\begingroup
\let\clearpage\relax

\section{Introduction}

Signals and systems often exhibit different structures at different scales. Such multi-scale structure has inspired a wide variety of multi-scale signal transforms, such as wavelets~\cite{Mallat:1989be}, curvelets~\cite{Candes:2006bs} and multi-scale pyramids~\cite{Simoncelli:1992fh}, that can represent natural signals compactly. Moreover, their ability to compress signal information into a few significant coefficients has made multi-scale signal transforms valuable beyond compression and are now commonly used in signal reconstruction applications, including denoising~\cite{Donoho:2006el}, compressed sensing~\cite{Donoho:2006ci, Candes:2006eq}, and signal separation~\cite{Donoho:2001hn,Starck:2005be,Chen:2006hm}. By now, multi-scale modeling is associated with many success stories in engineering applications.

On the other hand, low rank methods are commonly used instead when the signal subspace needs to be estimated as well. In particular, low rank methods have seen great success in \revised{applications, such as biomedical imaging~\cite{Liang:2007gf}, face recognition~\cite{Basri:2003ie} and collaborative filtering~\cite{Goldberg:2001hj}, that require} exploiting the global data correlation to recover the signal subspace and compactly represent the signal at the same time. Recent convex relaxation techniques~\cite{Fazel:2002tc} have further enabled low rank model to be adaptable to \revised{\erased{practical applications} various signal processing tasks}, including matrix completion~\cite{Candes:2009kj},  system identification~\cite{Fazel:2013dc} and phase retrieval~\cite{Candes:2013ka}, making low rank methods ever more attractive.

In this paper, we present a multi-scale low rank matrix decomposition method that incorporates multi-scale structures with low rank methods. The additional multi-scale structure allows us to obtain a more accurate and compact signal representation than conventional low rank methods whenever the signal exhibits multi-scale structures (see Figure~\ref{fig:hanning}).  To capture data correlation at multiple scales, we model our data matrix as a sum of block-wise low rank matrices with increasing scales of block sizes (more detail in Section~\ref{sec:model}) and consider the inverse problem of decomposing the matrix into its multi-scale components. Since we do not assume an explicit basis model, multi-scale low rank decomposition also prevents modeling errors or basis mismatch that are commonly seen with multi-scale signal transforms. In short, our proposed multi-scale low rank decomposition inherits the merits from both multi-scale modeling and low rank matrix decomposition.

\begin{figure*}[!ht]
\begin{center}

\ifdefined \single
\includegraphics[width=1.1\linewidth]{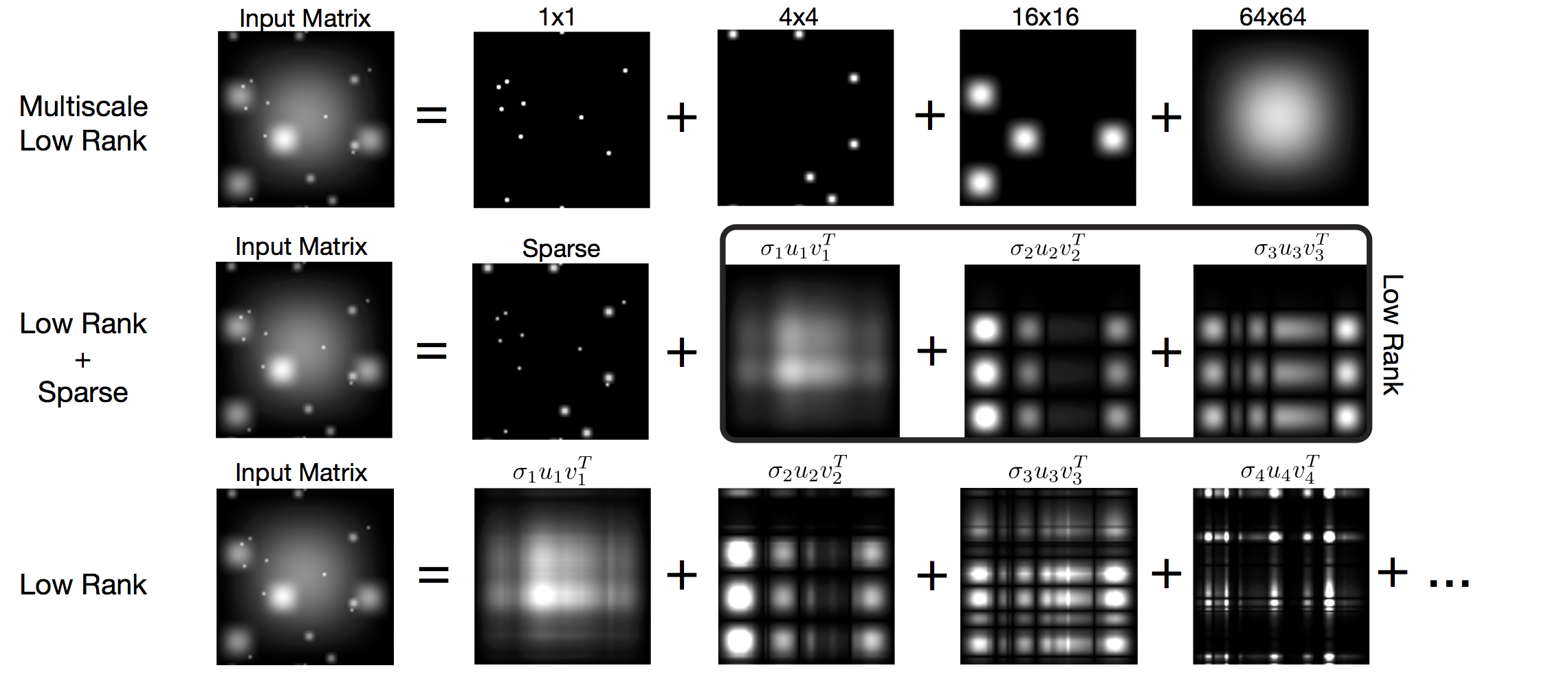}
\else
\includegraphics[width=0.9\linewidth]{figures/hanning.png}
\fi

\caption{An example of our proposed multi-scale low rank decomposition compared with other low rank methods. Each blob in the input matrix is a rank-1 matrix constructed from an outer product of hanning windows. Only the multi-scale low rank decomposition exactly separates the blobs to their corresponding scales and represents each blob as compactly as possible.}
\label{fig:hanning}
\end{center}
\end{figure*}

Leveraging recent convex relaxation techniques, we propose a convex formulation to perform the multi-scale low rank matrix decomposition. We provide a theoretical analysis in Section~\ref{sec:theory} that extends the rank-sparsity incoherence results in Chandrasekaran et al.~\cite{Anonymous:2011kn}. We show that  the proposed convex program decomposes the data matrix into its multi-scale components exactly \revised{under a deterministic incoherence condition}. \revised{In addition, in Section~\ref{sec:noisy}, we provide a theoretical analysis on approximate multi-scale low rank matrix decomposition in the presence of additive noise that extends the work of Agarwal et al.~\cite{Agarwal:2012gc}.}

A major component of this paper is to introduce the proposed multi-scale low rank decomposition with emphasis on its practical performance and applications. We provide practical guidance on choosing regularization parameters for the convex method in Section~\ref{sec:lambda} and describe heuristics to perform cycle spinning~\cite{Coifman:1995ji} to reduce blocking artifacts in Section~\ref{sec:ti}. In addition, we applied the multi-scale low rank decomposition on real datasets and considered four applications of the multi-scale low rank decomposition: illumination normalization for face images, motion separation for surveillance videos, compact modeling of the dynamic contrast enhanced magnetic resonance imaging and collaborative filtering exploiting age information. (See Section~\ref{sec:exp} for more detail). Our results show that the proposed multi-scale low rank decomposition provides intuitive multi-scale decomposition and compact signal representation for a wide range of applications.

\subsection*{ Related work }

Our proposed multi-scale low rank matrix decomposition draws many inspirations from recent developments in rank minimization~\cite{Fazel:2002tc,Recht:2010dx,Agarwal:2012gc,Xu:2010uc,Candes:2010jb,Candes:2009kj,Recht:2010ht, Hsu:2011di}. In particular, the multi-scale low rank matrix decomposition is a generalization of the low rank + sparse decomposition proposed by Chandrasekaran et al.~\cite{Anonymous:2011kn} and Cand{\`e}s et al.~\cite{Candes:2011bd}. Our multi-scale low rank convex formulation also fits into the convex demixing framework proposed by McCoy et al.~\cite{McCoy:2013vc,McCoy:2014ir,McCoy:ky}, who studied the problem of demixing components via convex optimization. The proposed multi-scale low rank decomposition can be viewed as a concrete and practical example of the convex demixing problem. \revisedtwo{\reviewercommentnumtwo{R1.3}However, their theoretical analysis assumes that each component is randomly oriented with respect to each other, and does not apply to our setting, where we observe the direct summation of the components.} Bakshi et al.~\cite{Bakshi:1998il} proposed a multi-scale principal component analysis by applying principal component analysis on wavelet transformed signals, but such method implicitly constrains the signal to lie on a predefined wavelet subspace. Various multi-resolution matrix factorization techniques~\cite{Kondor:2014vr,Kakarala:2001fk} were proposed to greedily peel off components of each scale by recursively applying matrix factorization. One disadvantage of these factorization methods is that it is not straightforward to incorporate them with other reconstruction problems as models. \revised{\reviewercommentnum{R3}Similar multi-scale modeling using demographic information was also used in collaborative filtering described in Vozalis and Margaritis~\cite{VOZALIS:2007ki}.}



\section{Multi-scale Low Rank Matrix Modeling}
\label{sec:model} 

\revised{\reviewercommentnum{R1.5}\erased{Low rank matrix modeling arises frequently in a wide variety of applications such as biomedical imaging, face recognition and collaborative filtering. In particular, when multiple copies of similar data $\{y_i\}_{i=1}^N$ are observed, the data matrix $Y$ constructed as follows is often low rank:}}

\revised{\erased{While low rank modeling captures the notion of data similarity, it completely ignores any locality information that may be present in the data matrix. For example, in video processing, each data vector $y_i$ represents a video frame and it is intuitive that each video frame should be more correlated with nearby frames than faraway frames. Hence the block matrix rank constructed from the data matrix is much lower than the global matrix rank. Such local low rank structure has been observed in various other applications, in particular in imaging applications. Since natural data are naturally correlated in multiple scales, a multi-scale low rank modeling is intuitively a more appropriate modeling.}}

\begin{figure}[!ht]
\begin{center}

\ifdefined \single
	\includegraphics[width=0.9\linewidth]{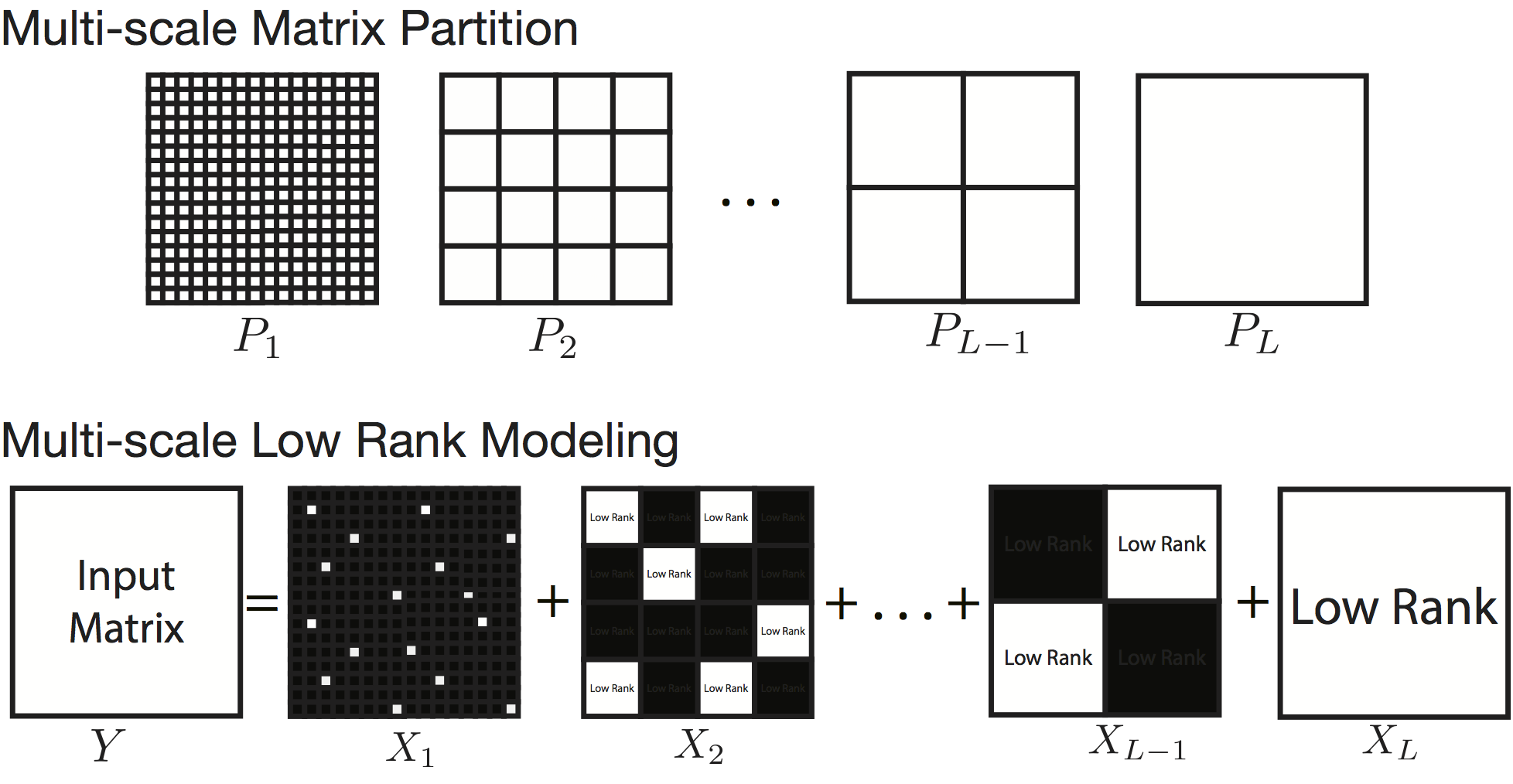}
\else
	\includegraphics[width=\linewidth]{figures/blocks.png}
\fi
\caption{Illustration of a multi-scale matrix partition and its associated multi-scale low rank modeling. Since the zero matrix is a matrix with the least rank, our multi-scale modeling naturally extends to sparse matrices as $1 \times 1$ low rank matrices.}
\label{fig:blocks}
\end{center}
\end{figure}

\revisedtwo{In this section, we describe the proposed multi-scale low rank matrix modeling in detail.} To concretely formulate the model, we assume that we can partition the data matrix of interest $Y$ into different scales. Specifically, we assume that we are given a multi-scale partition $\seti{ \part_i }$ of the indices of an $M \times N$ matrix, where each block $\block$ in $\part_i$ is an order magnitude larger than the blocks in the previous scale $\part_{i-1}$. Such multi-scale partition can be naturally obtained in many applications. For example in video processing, a multi-scale partition can be obtained by decimating both space and time dimensions. Figures~\ref{fig:blocks} and~\ref{fig:blocks2} provide two examples of a multi-scale partition, the first one with decimation along two dimensions and the second one with decimation along one dimension. In Section~\ref{sec:exp}, we provide practical examples on creating these multi-scale partitions for different applications.

To easily transform between the data matrix and the block matrices, we then consider a block reshape operator $R_\block (X)$  that extracts a block $\block$ from the full matrix $X$ and reshapes the block into an $m_i \times n_i$ matrix (Figure~\ref{fig:Rop}) \revised{\reviewercommentnum{R1.8}, where $m_i \times n_i$ is the $i$th scale block matrix size determined by the user}. 

\begin{figure}[!ht]
\begin{center}

\ifdefined \single
	\includegraphics[width=0.9\linewidth]{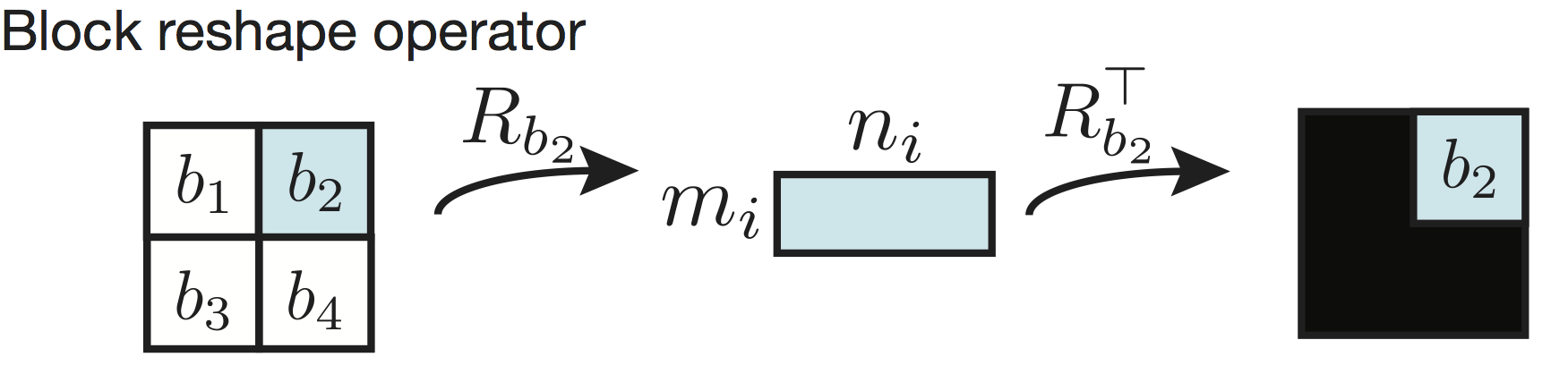}
\else
	\includegraphics[width=\linewidth]{figures/Rop.png}
\fi
\caption{Illustration of the block reshape operator $R_{\block}$. $R_\block$ extracts block $\block$ from the full matrix and reshapes it into an $m_i \times n_i$ matrix. Its adjoint operator $R_\block^\top$ takes an $m_i \times n_i$ matrix and embeds it into a full-size zero matrix.}
\label{fig:Rop}
\end{center}
\end{figure}

Given an $M \times N$ input matrix $Y$ and its corresponding multi-scale partition and block reshape operators, we propose a multi-scale low rank modeling that models the $M \times N$ input matrix $Y$ as a sum of matrices $\sumi{X_i}$, in which each $X_i$ is block-wise low rank with respect to its partition $P_i$, that is,
\begin{align*}
				Y  &= \sumi X_i \\
				X_i &=  \sumb \Rb^\top(  \Ub \Sb \VTb )
\end{align*}
 where $\Ub$, $\Sb$, and $\Vb$ are matrices with sizes $m_i \times \rank_b$, $\rank_b \times \rank_b$ and $n_i \times \rank_b$ respectively and form the rank-$\rank_b$ reduced SVD of $\Rb( X_i )$. Note that when the rank of the block matrix $\Rb({ X_i })$ is zero, we have $ \{\Ub , \Sb, \Vb \}$ as empty matrices, which do not contribute to $X_i$. Figure~\ref{fig:blocks} and \ref{fig:blocks2} provide illustrations of two kinds of modeling with their associated partitions.

\begin{figure}[!ht]
\begin{center}

\ifdefined \single
	\includegraphics[width=0.9\linewidth]{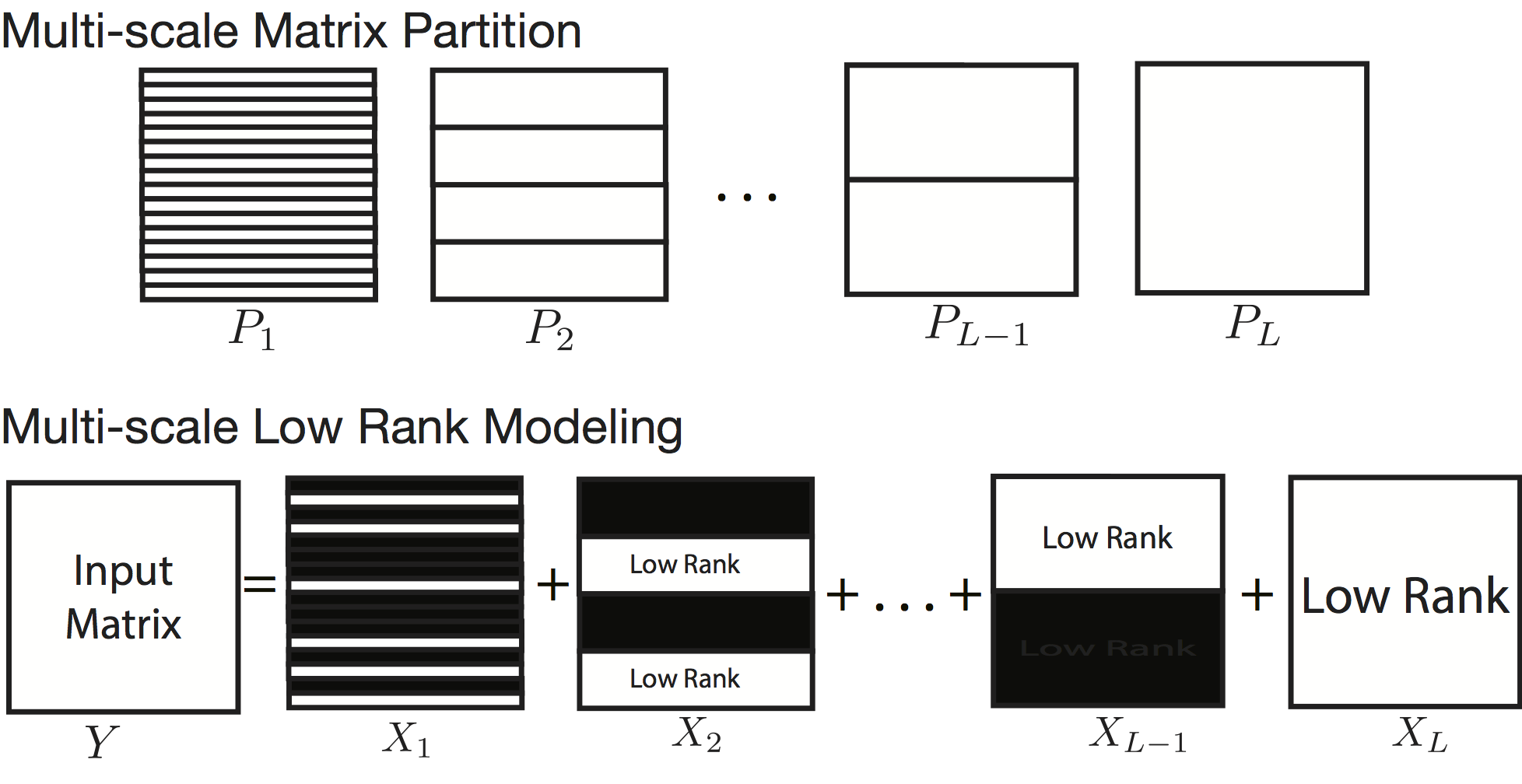}
\else
	\includegraphics[width=\linewidth]{figures/blocks2.png}
\fi

\caption{Illustration of another multi-scale matrix partition and its associated multi-scale low rank modeling. Here, only the vertical dimension of the matrix is decimated. Since a $1 \times N$ matrix is low rank if and only if it is zero, our multi-scale modeling naturally extends to group sparse matrices.}
\label{fig:blocks2}
\end{center}
\end{figure}

By constraining each block matrices to be of low rank, the multi-scale low rank modeling captures the notion that some nearby entries are more similar to each other than global entries in the data matrix. We note that the multi-scale low rank modeling is a generalization of the low rank + sparse modeling proposed by Chandrasekaren et al.~\cite{Anonymous:2011kn} and Cand{\`e}s et al.~\cite{Candes:2011bd}. In particular, the low rank + sparse modeling can be viewed as a $2$-scale low rank modeling, in which the first scale has block size $1 \times 1$ and the second scale has block  size $M \times N$. By adding additional scales between the sparse and globally low rank matrices, the multi-scale low rank modeling can capture locally low rank components that would otherwise need many coefficients to represent for low rank + sparse.

Given a data matrix $Y$ that fits our multi-scale low rank model, our goal is to decompose the data matrix $Y$ to its multi-scale components $\seti{X_i}$. The ability to recover these multi-scale components is beneficial for many applications and allows us to, for example, extract motions at multiple scales in surveillance videos (Section~\ref{sec:exp}). Since there are many more parameters to be estimated than the number of observations, it is necessary to impose conditions on $X_i$. In particular, we will exploit the fact that each block matrix is low rank via a convex program, which will be described in detail in section~\ref{sec:cvx}.

\subsection{Multi-scale low rank + noise}
\label{ssec:robust}
Before moving to the convex formulation, we note that our multi-scale matrix modeling can easily account for data matrices that are corrupted by additive white Gaussian noise. Under the multi-scale low rank modeling, we can think of the additive noise matrix as the largest scale signal component and is unstructured in any local scales. Specifically if we observe instead the following
\begin{align*}
		Y = \sumi X_i + X_{\noise} 
		\numberthis \label{noise}
\end{align*}
where $X_{\noise}$ is an independent  and identically distributed Gaussian noise matrix. Then we can define a reshape operator $R_{\noise}$ that reshapes the entire matrix into an $MN \times 1$ vector and the resulting matrix fits exactly to our multi-scale low rank model with $L+1$ scales. This incorporation of noise makes our model flexible in that it automatically provides a corresponding convex relaxation, a regularization parameter for the noise matrix and allows us to utilize the same iterative algorithm to solve for the noise matrix. Figure~\ref{fig:hanning_noisy} provides an example of the noisy multi-scale low rank matrix decomposition.

\begin{figure}[!ht]
\begin{center}

\ifdefined \single
	\includegraphics[width=0.9\linewidth]{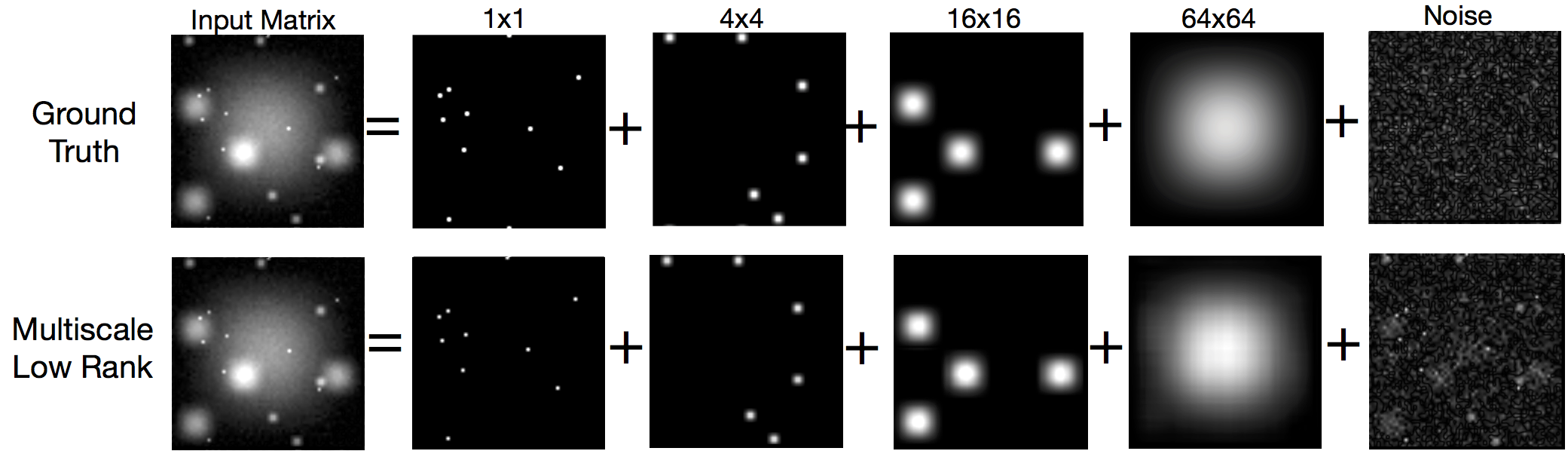}
\else
	\includegraphics[width=\linewidth]{figures/hanning_noisy.png}
\fi

\caption{An example of the multi-scale low rank decomposition in the presense of additive Gaussian noise by solving the convex program~\eqref{eq:cvx}.}
\label{fig:hanning_noisy}
\end{center}
\end{figure}

\section{Problem Formulation and Convex Relaxation}
\label{sec:cvx}

Given a data matrix $Y$ that fits the multi-scale low rank model, our goal is to recover the underlying multi-scale components $\seti{X_i}$ using the fact that $X_i$ is block-wise low rank. Ideally, we would like to obtain a multi-scale decomposition with the minimal block matrix rank and solve a problem similar to the following form:
\begin{equation*}
	\begin{aligned}
		& \underset{X_1,\hdots ,X_\levels }{\text{minimize}}
		& &  \sumi \sumb  \text{ rank} (\Rb ( X_i ) ) \\
		& \text{subject to}
		& & Y = \sumi X_i
	\end{aligned}
\end{equation*}

However, each rank minimization for each block is combinatorial in nature. In addition, it is not obvious whether the direct summation of ranks is a correct formulation as a $1$-sparse matrix and a rank-$1$ matrix should intuitively not carry the same cost. Hence, the above non-convex problem is not a practical formulation to obtain the multi-scale decomposition.

Recent development in convex relaxations suggests that rank minimization problems can often be relaxed to a convex program via nuclear norm relaxation~\cite{Fazel:2002tc, Recht:2010ht}, while still recovering the optimal solution to the original problem. In particular, Chandrasekaren et al.~\cite{Anonymous:2011kn} and Cand{\`e}s et al.,~\cite{Candes:2011bd} showed that a low rank + sparse decomposition can be relaxed to a convex program by minimizing a nuclear norm + $\ell 1$-norm objective as long as the signal constituents are incoherent with respect to each other. In addition, Cand{\`e}s et al.,~\cite{Candes:2011bd} showed that the regularization parameters for sparsity and low rank should be related by the square root of the matrix size. Hence, there is hope that, along the same line, we can perform the multi-scale low rank decomposition exactly via a convex formulation.

Concretely, let us define $\nuc{\cdot} $ to be the nuclear norm, the sum of singular values, and $\msv{\cdot} $ be the maximum singular value norm. For each scale $i$, we consider the block-wise nuclear norm to be the convex surrogate for the block-wise ranks and define $\normi{ \cdot }{i}$ the block-wise nuclear norm for the $i$th scale as
\begin{align*}
		\normi{ \cdot} {i}  = \sum_{\block \in \part_i}  \nuc{ R_{\block} ( \cdot ) } 
\end{align*}
Its associated dual norm $\| \cdot \|_{(i)}^*$ is then given by	
\begin{align*}
		\| \cdot \|_{(i)}^* = \max_{\block \in \part_i}  \msv{ R_{\block} ( \cdot ) }
\end{align*}
which is the maximum of all block-wise maximum singular values.

We then consider the following convex relaxation for our multi-scale low rank decomposition problem:
\begin{equation}
	\begin{aligned}
		& \underset{X_1,\hdots ,X_\levels }{\text{minimize}}
		& &  \sumi \lambda_i \| X_i \|_{(i)}\\
		& \text{subject to}
		& & Y = \sumi X_i
	\end{aligned}
	\label{eq:cvx}
\end{equation}
where $\seti{ \lambda_i }$ are the regularization parameters and their selection will be described in detail in section~\ref{sec:lambda}.

Our convex formulation is a natural generalization of the low rank + sparse convex formulation~\cite{Anonymous:2011kn,Candes:2011bd}. With the two sided matrix partition (Fig.~\ref{fig:blocks}), the nuclear norm applied to the $1 \times 1$ blocks becomes the element-wise $\ell 1$-norm and the norm for the largest scale is the nuclear norm. With the one sided matrix partition (Fig.~\ref{fig:blocks2}), the nuclear norm applied to $1 \times N$ blocks becomes the group-sparse norm and can be seen as a generalization of the group sparse + low rank decomposition~\cite{Xu:2010uc}. If we incorporate additive Gaussian noise in our model as described in Section~\ref{ssec:robust}, then we have a nuclear norm applied to an $MN \times 1$ vector, which is equivalent to the Frobenius norm. 

One should hope that the theoretical conditions from low rank + sparse can be generalized rather seamlessly to the multi-scale counterpart. Indeed, in Section~\ref{sec:theory}, we show that the core theoretical guarantees in the work of Chandrasekaren et al.~\cite{Anonymous:2011kn} on exact low rank + sparse decomposition can be generalized to the multi-scale setting. \revised{In section~\ref{sec:noisy}, we show that the core theoretical guarantees in the work of Agarwal et al.~\cite{Agarwal:2012gc} on noisy matrix decomposition can be generalized to the multi-scale setting as well to provide approximate decomposition guarentees.} \revised{\erased{At a high level, we can show that as long as the row and column spaces of each block matrices are not coherent, then we can select appropriate regularization parameters $\seti{\lambda_i}$ such that the resulting convex program recovers the multi-scale signal components $\seti{X_i}$ from the data matrix $Y$.}}


\revisedtwo{\erasedtwo{Theoretical justification of the multi-scale low rank convex formulation can also be obtained via the convex demixing framework introduced in McCoy et al. In their work, McCoy et al. studied the problem of demixing signal constituents using a convex program and modeled the signal components to be randomly oriented with respect to each other. While the random orientation assumption is not satisfied in our setting \revised{as we observe the summation of the signal components}, McCoy et al. provided strong guarantees on when the demixing succeeds and predicted a phase transition phenomenon. Since the multi-scale low rank decomposition fits the convex demixing framework, it also enjoys the same theoretical guarantees whenever the random orientation approximation is appropriate.}}

\section{Guidance on Choosing Regularization Parameters }
\label{sec:lambda}

In this section, we provide \revised{practical} guidance on selecting the regularization parameters $\seti{\lambda_i}$. Selecting the regularization parameters $\seti{\lambda_i}$ is crucial for the convex decomposition to succeed, both theoretically and practically. While theoretically we can establish criteria on selecting the regularization parameters (see Section~\ref{sec:theory} and \ref{sec:noisy}), such parameters are not straightforward to calculate in practice as it requires properties of the signal components $\seti{X_i}$  \em before \em the decomposition.

To select the regularization parameters $\seti{\lambda_i}$ \revised{in practice}, we follow the suggestions from Wright et al.~\cite{Anonymous:2013fx} and Fogel et al.~\cite{Foygel:cl}, and set each regularization parameter $\lambda_i$ to be the Gaussian complexity of each norm $ \normi{\cdot}{i}$, which is defined as the expectation of the dual norm of random Gaussian matrix:
\begin{equation}
	\begin{aligned}
	\lambda_i \sim E [ \| G \|_{(i)}^* ] 
	\end{aligned}
	\label{eq:gauss}
\end{equation}
where  $\sim$ denotes equality up to some constant and $G$ is a unit-variance independent and identically distributed random Gaussian matrix.

The resulting expression for the Gaussian complexity is the maximum singular value of a random Gaussian matrix, which has been studied extensively \revised{\erased{in}by} Bandeira and Handel~\cite{Bandeira:2014wk}. The recommended regularization parameter for scale $i$ is given by
\begin{equation}
\boxed{
	\begin{aligned}
	\lambda_i \sim \sqrt{ m_i }  + \sqrt{ n_i } + \sqrt{ \log \left( \frac{ M N } { \max  (m_i, n_i) }  \right)}
	\end{aligned}
	\label{eq:lambda}
	}
\end{equation}
For the sparse matrix scale with $1 \times 1$ block size, $\lambda_i \sim \sqrt{ \log  (   M N ) }$ and for the globally low rank scale with $M \times N$ block size, $\lambda_i \sim \sqrt{ M }  + \sqrt{ N }  $. Hence this regularization parameter selection is consistent with the ones recommended for low rank + sparse decomposition by Cand{\`e}s et al.~\cite{Candes:2011bd}, \revised{up to a $\log$ factor}. In addition, for the noise matrix with $MN \times 1$ block size, $\lambda_i \sim \sqrt{ M N }$, which has similar scaling as in square root LASSO~\cite{Belloni:2011wd}. In practice, we found that the suggested regularization parameter selection allows exact multi-scale decomposition when the signal model is matched (for example Figure~\ref{fig:hanning}) and provides visually intuitive decomposition for real datasets.

\revised{ 
For approximate multi-scale low rank decomposition in the presence of additive noise, some form of theoretical guarantees for the regularization selection can be found in our analysis in Section~\ref{sec:noisy}. In particular, we show that if the regularization parameters $\lambda_i$ is larger than the Gaussian complexity of $\dnormi{\cdot}{i}$ in addition to some ``spikiness" parameters, then the error between recovered decomposition and the ground truth $\seti{X_i}$ is bounded by the block-wise matrix rank.}

\revised{\erased{ one explanation for the above selection of the regularization parameters can be seen from the iterative procedure described in the Section~\ref{sec:solve}. As described in the iterative algorithm, the regularization parameters are the thresholds that are chosen to suppress interference from other scales. Now consider the special case that our input matrix $Y$ is a random Gaussian matrix. Then in order to have exact recovery, we would want to have all multi-scale components $\seti{X_i}$ to be zero except the one that represents the noise matrix. Hence, each threshold or regularization parameter, should be chosen such that it is greater than all the block-wise singular values of the Gaussian matrix. That is,
$\lambda_i  \ge \max_{\block \in \part_i}  \msv{ R_{\block} ( G ) }$
which is simply the dual norm of the Gaussian matrix. By results in concentration inequalities, we know that choosing $\lambda_i$ to be the expectation of the dual norm guarantees that the inequality is satisfied with high probability.}}



\section{Theoretical Analysis for Exact Decomposition}
\label{sec:theory}

In this section, we provide a theoretical analysis of the proposed convex formulation and show that if $\seti{ X_i}$ satisfies a deterministic incoherence condition, then the proposed convex formulation (\ref{eq:cvx}) recovers $\seti{X_i}$ from $Y$ exactly. 

Our analysis follows similar arguments taken by Chandrasekaren et al.~\cite{Anonymous:2011kn} on low rank + sparse decomposition and generalizes them to the proposed multi-scale low rank decomposition. Before showing our main result (Theorem~\ref{thm:main}), we first describe the subgradients of our objective function (Section \ref{ssec:sub}) and define a coherence parameter in terms of the block-wise row and column spaces (Section \ref{ssec:inc}).

\subsection{ Subdifferentials of the block-wise nuclear norms }
\label{ssec:sub}

To characterize the optimality of our convex problem, we first look at the subgradients of our objective function. We recall that for any matrix $X$ with $\{U, S, V \}$ as its reduced SVD representation, the subdifferential of $\nuc{\cdot}$ at $X$ is given by~\cite{Recht:2010ht, Watson:1992gq},
\begin{align*}
		 \partial \nuc{ X } &= \left \{     U V^\top + W : W \text{~and~} X \text{ have orthogonal row} \right . \\
  & \left . \phantom{\sum_d} \text{ and column spaces and~} \msv{ W } \le 1    \right \}
\end{align*}
Now recall that we define the block-wise nuclear norm to be $\normi{ \cdot} {i}  = \sum_{\block \in \part_i}  \nuc{ R_{\block} ( \cdot ) } $. Then using the chain rule and the fact that $R_\block(X_i) = \Ub \Sb \Vb^\top$, we obtain an expression for the subdifferential of $\normi{ \cdot } {i}$ at $X_i$ as follows:
\begin{align*}
		& \partial \normi{ X_i } {i}  = \left\{  \sumb  R_{\block}^\top (U_{\block} V_{\block}^\top + W_{\block}) : W_\block \text{~and~} R_{\block} (X_i) \text{ have }  \right . \\
                 & \left . \phantom{\sum_d}  \text{ orthogonal row and column spaces and~} \msv{ W_\block } \le 1    \right\} 
\end{align*}

To simplify our notation, we define $\atom_i  =  \sumb R_{\block}^\top(   U_{\block}  V_{\block}^\top  )$ and $T_i$ to be a vector space that contains matrices with the same block-wise row spaces or column spaces as $X_i$, that is,
\begin{align*}
		 T_i  &= \left\{   \sumb R_{\block}^\top(   U_{\block} X_{\block}^\top + Y_{\block} V_{\block}^\top  ): X_\block \in \complex^{n_i \times \rank_i}, Y_\block \in \complex^{ m_i \times \rank_i } \right\}
\end{align*}
\revised{\reviewercommentnum{1.8}where $m_i \times n_i$ is the size of the block matrices for scale $i$ and $\rank_b$ is the matrix rank for block $b$.}
Then, the subdifferential of each $\normi{ \cdot }{i}$ at $X_i$ can be compactly represented as,
\begin{align*}
		 \partial \normi{ X_i } {i}  = \left \{  \atom_i + W_{i} : W_i \in T^\perp_i  \text{~and~}  \dnormi{W_i}{i}  \le 1 \right \}
\end{align*}
\revised{We note that $E_i$ can be thought of as the ``sign" of the matrix $X_i$, pointing toward the principal components, and, in the case of the sparse scale, is exactly the sign of the entries.}

In the rest of the section, we will be interested in projecting a matrix $X$ onto $T_i$, which can be performed with the following operation:
\begin{align*}
		\proj_{T_i}(X)  = &\sumb  R_b^\top \left ( \Ub \Ub^\top R_b(X) + R_b(X) \Vb \Vb^\top \right .\\
		& \left . -  \Ub \Ub^\top R_b(X) \Vb \Vb^\top  \right )
\end{align*}
Similarly, to project a matrix $X$ onto the orthogonal complement of $T_i$, we can apply the following operation:
\begin{align*}
		\proj_{T_i^\perp}(X)  &=  \sumb   R_b^\top\left( (I - \Ub \Ub^\top) R_b(X) ( I - \Vb \Vb^\top )\right)
\end{align*}
where $I$ is an appropriately sized identity matrix.

\subsection{Incoherence}
\label{ssec:inc}

\revisedtwo{\erasedtwo{Recent developments in signal processing suggest that exact signal decomposition can be achieved as long as the components are incoherent with respect to each other.}} Following Chandrasekaren et al.~\cite{Anonymous:2011kn}, we consider a deterministic measure of incoherence through the block-wise column and row spaces of $X_i$. Concretely, we define the coherence parameter for the $j$th scale signal component $X_j$ with respect to the $i$th scale to be the following:
\begin{align*}
		\mu_{ij} =   \max_{  N_j \in T_j ,~ \dnormi{N_j}{j} \le 1 }  \dnormi{ N_j } {i}
	\numberthis
	\label{eq:inc_def}
\end{align*}

Using $\mu_{ij}$ as a measure of incoherence, we can quantitatively say that the $j$th scale signal component is incoherent with respect to the $i$th scale if $\mu_{ij}$ is small. In the case of low rank + sparse, Chandrasekaren et al.~\cite{Anonymous:2011kn} provides excellent description of the concepts behind the coherence parameters. We refer the reader to their paper for more detail.

\subsection{Main Result}
\label{ssec:thm}

Given the above definition of incoherence, the following theorem states our main result for exact multi-scale low rank decomposition:

\begin{theorem}\label{thm:main}

If we can choose regularization parameters $\seti{ \lambda_i }$ such that
\begin{align*}
			\sum_{j \ne i}  \mu_{ij}   \frac{\lambda_j}{ \lambda_i} < \frac{1}{2} , ~~~ \text{for } i = 1, \hdots, L
\end{align*}
then $\{ X_i \}_{i=1}^L$ is the unique optimizer of the proposed convex problem~(\ref{eq:cvx}).

In particular when the number of scales $L = 2$, the condition on $\{ \mu_{12}, \mu_{21} \}$ reduces to $\mu_{12} \mu_{21} < 1/4$ and the condition on $\{\lambda_1, \lambda_2\}$ reduces to $2 \mu_{12} <  \lambda_1 / \lambda_2   <  1 / (2 \mu_{21})$, which is in similar form as Theorem 2 in Chandrasekaren et al.~\cite{Anonymous:2011kn}.

\end{theorem}

\revisedtwo{The proof for the above theorem is given in Appendix~\ref{appendix:exact}.}

\section{\revised{Theoretical Analysis for Approximate Decomposition}}
\label{sec:noisy}

In this section, we provide a theoretical analysis for approximate multi-scale low rank decomposition when the measurement is corrupted by additive noise as described in Section~\ref{ssec:robust}. Our result follows arguments from Agarwal et al.~\cite{Agarwal:2012gc} on noisy $2$-scale matrix decomposition and extends it to the multi-scale setting. 

Instead of using the incoherence parameter $\mu_{ij}$ defined for the exact decomposition analysis in Section~\ref{sec:theory}, we opt for a weaker characterization of incoherence between scales for approximate decomposition, studied in Agarwal et al.~\cite{Agarwal:2012gc}. Concretely, we consider spikiness parameters $\alpha_{ij}$ between the $j$th signal component $X_j$ and $i$th scale norm $\normi{\cdot}{i}$ such that,
\begin{align*}
 \alpha_{ij} = \dnormi{ X_j }{i}
\end{align*}
for each $j \ne i$. Hence, if $\alpha_{ij}$ is small, we say $X_j$ is not spiky with respect to the $i$th norm.

For analysis purpose, we also impose the constraints $\dnormi{ X_j }{i} \le \alpha_{ij}$ in the convex program. That is, we consider the solution from the following convex program:
\begin{align*}
		& \underset{X_1,\hdots ,X_\levels, X_\noise }{\text{minimize}} \numberthis \label{eq:noisy}
		& &   \sumi \lambda_i \normi{ X_i }{i} + \lambda_\noise \fro{ X_\noise } \\
		& \text{subject to}
		& & Y = \sumi X_i + X_\noise\\
		& & & \dnormi{ X_j }{i} \le \alpha_{ij} ~~~~\text{for $j \ne i$} \numberthis \label{eq:constraint}
\end{align*}
We emphasize that the additional constraints~\eqref{eq:constraint} are imposed only for the purpose of theoretical analysis and are not imposed in our experimental results. \revisedtwo{In particular, for our simulation example in Figure~\ref{fig:hanning_noisy}, the minimizer of the convex program~\eqref{eq:cvx}, using the recommended regularization parameters in Section~\ref{sec:lambda}, satisfied the constraints~\eqref{eq:constraint} even when the constraints were not imposed.}

\revisedtwo{Let us define $\seti{\Delta_i}$ and $\Delta_\noise$ to be the errors between the ground truth components $\seti{X_i}$ and $X_\noise$ and the minimizers of convex program $\eqref{eq:noisy}$}. Then, equivalently, we can denote  $\seti{ X_i + \Delta_i }$ and $X_\noise + \Delta_\noise$ as the minimizers of the convex program~\eqref{eq:noisy}. The following theorem states our main result for approximate decomposition.

\begin{theorem}
If we choose $\seti{\lambda_i}$ such that
\begin{align*}
		  \lambda_i &\ge  2 \lambda_\noise \frac{ \dnormi{ X_\noise }{i} } { \fro{ X_\noise }} + \sum_{j \ne i} 2\alpha_{ij}
		  \numberthis \label{lambda_condition}
\end{align*}
and $\lambda_\noise$ such that
\begin{align*}
		 \lambda_\noise \ge \sqrt{64 \sumi \lambda_i^2  \sumb \rank_b}
		  \numberthis \label{lambdaN_condition}
\end{align*}
then the error is bounded by
\begin{align*}
		  \sumi \fro{  \Delta_i }
		&\lesssim  \frac{\fro{ X_\noise } }{\lambda_\noise} \sumi \lambda_i  \sqrt{\sumb \rank_b }
\end{align*}
where $\lesssim$ denotes inequality up to a universal constant.
\end{theorem}
\label{thm:noisy}
Hence, when the spikiness parameters are negligible and $X_\noise = \sigma G$, where $G$ is an independent, identically distributed Gaussian noise matrix with unit variance and $\sigma$ is the noise standard deviation, choosing $\lambda_\noise \sim E[\fro{   G }] \sim \sqrt{MN}$ and $\lambda_i \sim  E[\dnormi{ G }{i} ] \sim \sqrt{m_i} + \sqrt{n_i} +\sqrt{ \log (  M N /\max  ( m_i, n_i ) )} $ ensures the condition is satisfied with high probability. This motivates the recommended regularization selection in Section~\ref{sec:lambda}.

The proof for the above theorem is given in Appendix~\ref{appendix:approximate} and follows arguments from Agarwal et al.~\cite{Agarwal:2012gc} on noisy matrix decomposition and Belloni et al.~\cite{Belloni:2011wd} on square root LASSO.


\section{An Iterative Algorithm for Solving the Multi-scale Low Rank Decomposition}
\label{sec:solve}

In the following, we will derive an iterative algorithm that solves for the multi-scale low rank decomposition via the Alternating Direction of Multiple Multipliers (ADMM)~\cite{Boyd:2011bw}. While the proposed convex formulation (\ref{eq:cvx}) can be formulated into a semi-definite program, first-order iterative methods are commonly used when solving for large datasets for their computational efficiency and scalability. A conceptual illustration of the algorithm is shown in Figure~\ref{fig:algorithm}. \revisedtwo{\erasedtwo{Moreover, iterative algorithms often provide practical insights to the convex problem itself. For multi-scale low rank decomposition, our iterative algorithm has the interpretation of repeatedly removing inter-scale interference on the observed data via thresholding the singular values as shown in Figure~\ref{fig:algorithm}.}}

\begin{figure}[!ht]
\begin{center}

\ifdefined \single
	\includegraphics[width=0.8\linewidth]{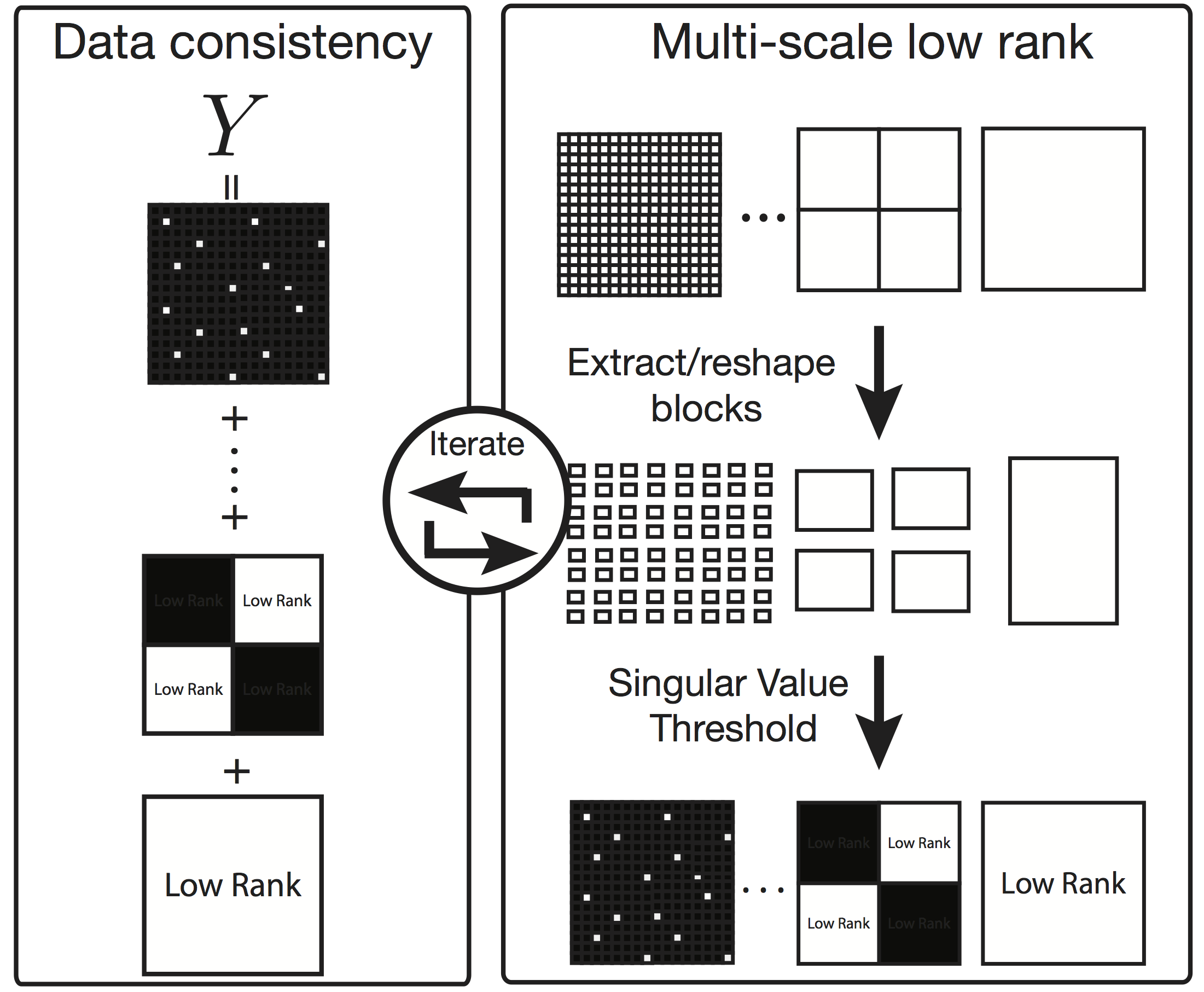}
\else
	\includegraphics[width=0.8\linewidth]{figures/algorithm.png}
\fi
\caption{A conceptual illustration of how to obtain a multi-scale low rank decomposition. First, we extract each block from the input matrix and perform a thresholding operation on its singular value to recover the significant components. Then, we subtract these significant components from our input matrix, thereby enabling the recovery of weaker, previously submerged components.}
\label{fig:algorithm}
\end{center}
\end{figure}

To formally obtain update steps using ADMM, we first formulate the problem into the standard ADMM form with two separable objectives connected by an equality constraint,
\begin{equation}
	\begin{aligned}
		& \underset{X_i, ~Z_i }{\text{minimize}}
		& & \mathbf{I} \left \{ Y = \sumi X_i  \right \} + \sumi \lambda_i \normi{ Z_i} {i}  \\
		& \text{subject to}
		& & { X_i = Z_i }
	\end{aligned}
	\label{eq:cvx_admm}
\end{equation}
where $\mathbf{I} \{ \cdot \}$ is the indicator function. 

To proceed, we then need to obtain the proximal operators~\cite{Parikh:2013vb} for the two objective functions $\mathbf{I}\{ Y = \sumi X_i  \} $ and $\sumi \lambda_i \normi{ Z_i} {i}$. For the data consistency objective $\mathbf{I}\{ Y = \sumi X_i  \} $, the proximal operator is simply the projection operator to the set. To obtain the proximal operator for the multi-scale nuclear norm objective $\sumi \lambda_i \normi{ X_i} {i}$, we first recall that the proximal operator for the nuclear norm $\nuc{ X }$ with parameter $\lambda$ is given by the singular value soft-threshold operator~\cite{Recht:2010ht},
\begin{equation}
	\begin{aligned}
		\svt_{\lambda} ( X ) = U \max (\Sigma - \lambda , 0 )V^\top
	\end{aligned}
	\label{eq:svt}
\end{equation}

Since we defined the block-wise nuclear norm for each scale $i$ as $\sum_{\block \in \part_i}  \nuc{ R_{\block} ( \cdot ) } $, the norm is separable with respect to each block and its proximal function with parameter $\lambda_i$ is given by the block-wise singular value soft-threshold operator,
\begin{equation}
	\begin{aligned}
		\bsvt_{\lambda_i} ( X ) = \sum_{\block \in P_i} R_{\block}^\top ( \svt_{\lambda_i} ( R_{\block} (X) ) ) 
	\end{aligned}
	\label{eq:bsvt}
\end{equation}
which simply extracts every blocks in the matrix, performs singular value thresholding and puts the blocks back to the matrix. We note that for $1 \times 1$ blocks, the block-wise singular value soft-threshold operator reduces to the element-wise soft-threshold operator and for $1 \times N$ blocks, the block-wise singular soft-threshold operator reduces to the joint soft-threshold operator.

Putting everything together and invoking the ADMM recipe~\cite{Boyd:2011bw}, we have the following algorithm to solve our convex multi-scale low rank decomposition (\ref{eq:cvx}):
\begin{equation}
\boxed{
	\begin{aligned}
		X_i & \leftarrow (Z_i - U_i) +  \frac{1}{L} \left ( Y - \sumi (Z_i - U_i) \right )  \\
		Z_i & \vphantom{\frac11}\leftarrow \bsvt_{\lambda_i / \rho} \left ( X_i + U_i \right ) \\
		U_i & \vphantom{\frac11}\leftarrow U_i - (Z_i - X_i)
	\end{aligned}
	\label{eq:admm}
	}
\end{equation}
where $\rho$ is the ADMM parameter that only affects the convergence rate of the algorithm.

The resulting ADMM update steps are similar in essence to the intuitive update steps in Figure~\ref{fig:algorithm}, and alternates between data consistency and enforcing multi-scale low rank. The major difference of ADMM is that it adds a dual update step with $U_i$, which bridges the two objectives and ensures the convergence to the optimal solution. Under the guarantees of ADMM, in the limit of iterations, $X_i$ and $Z_i$ converge to the optimal solution of the convex program (\ref{eq:cvx}) and $U_i$ converges to a scaled version of the dual variable. In practice, we found that $\sim 1000$ iterations are sufficient without any visible change for imaging applications. \revised{\reviewercommentnum{R1.4} Finally, we note that because the proximal operator for the multi-scale nuclear norm is computationally simple, other proximal operator based algorithms~\cite{Parikh:2013vb} can also be used.}

\section{ Computational Complexity }
\label{sec:speed}

Given the iterative algorithm~(\ref{eq:admm}), one concern about the multi-scale low rank decomposition might be that it is significantly more computationally intensive than other low rank methods as we have many more SVD's and variables to compute for. In this section, we show that because we decimate the matrices at each scale geometrically, the theoretical computational complexity of the multi-scale low rank decomposition is similar to other low rank decomposition methods, such as the low rank + sparse decomposition.

For concreteness, let us consider the multi-scale partition with two-sided decimation shown in Figure~\ref{fig:blocks} and have block sizes $m_i = 2^{i-1}$ and $n_i = 2^{i-1}$. Similar to other low rank methods, the SVD's dominate the per iteration complexity for the multi-scale low rank decomposition. For an $M \times N$ matrix, each SVD costs \revised{$\nflops( M \times N \text{ SVD} ) = O(MN^2)$}. The per iteration complexity for the multi-scale low rank decomposition is dominated by the summation of all the SVD's performed for each scale, which is given by,
\begin{equation}
	\begin{aligned}
		& \nflops \left( M \times N \text{ SVD} \right) + 4 ~ \nflops \left( M/2  \times N/2  \text{ SVD} \right) + \hdots  \\
		&= O( M N^2 ) +  O( M N^2 ) /2  +  O ( M N^2 ) /4 + \hdots  \\
		& \le  2 O(M N^2) \approx \nflops( M \times N \text{ SVD} )
	\end{aligned}
\end{equation}

Hence, the per-iteration computational complexity of the multi-scale low rank with two-sided decimated partition is on the order of a $M \times N$ matrix SVD. In general, one can show that the per-iteration complexity for arbitrary multi-scale partition is at most $\log( N )$ times the full matrix SVD.\revisedtwo{\erasedtwo{, which is tolerable if a multi-scale modeling is more suitable modeling.}}

While theoretically, the computation cost for small block sizes should be less than bigger block sizes, we found that in practice the computation cost for computing the small SVD's can dominate the per-iteration computation. This is due to the overhead of copying small block matrices and calling library functions repeatedly to compute the SVD's. 

Since we are interested in thresholding the singular values and in practice many of the small block matrices are zero as shown in Section~\ref{sec:exp}, one trick of reducing the computation time is to quickly compute an upper bound on the maximum singular value for block matrices before the SVD's. Then if the upper bound for the maximum singular value is less than the threshold, we know the thresholded matrix will be zero and can avoid computing the SVD.\revisedtwo{\erasedtwo{ at all.}} Since for any matrix $X$, its maximum singular value is bounded by the square root of any matrix norm on $X^\top X$ ~\cite{Horn:2012tf}, there are many different upper bounds that we can use. In particular, we choose the maximum row norm and consider the following upper bound,
\begin{equation}
	\begin{aligned}
		\sigma_{\max} ( X ) \le \sqrt{ \max_i \sum_j | X_{ik} X_{jk} | }
	\end{aligned}
\end{equation}

Using this upper bound, we can identify many below-the-threshold matrices before computing the SVD's at all. In practice, we found that the above trick provides a modest speedup of $3 \sim 5 \times$.

\section{Heuristics for translation invariant decomposition}
\label{sec:ti}

Similar to wavelet transforms, one drawback of the multi-scale low rank decomposition is that it is not translation invariant, \revised{\reviewercommentnum{R1.7}that is, shifting the input changes the resulting decomposition}. In practice, this translation variant nature often creates blocking artifacts near the block boundaries, which can be visually jarring for image or video applications. One solution to remove these artifacts is to introduce overlapping partitions of the matrix so that the overall algorithm is translation invariant. However, this vastly increases both memory and computation especially for large block sizes. In the following, we will describe a cycle spinning approach that we used in practice to reduce the blocking artifacts with only slight increase in per-iteration computation.

\begin{figure}[!ht]
\begin{center}

\ifdefined \single
	\includegraphics[width=0.9\linewidth]{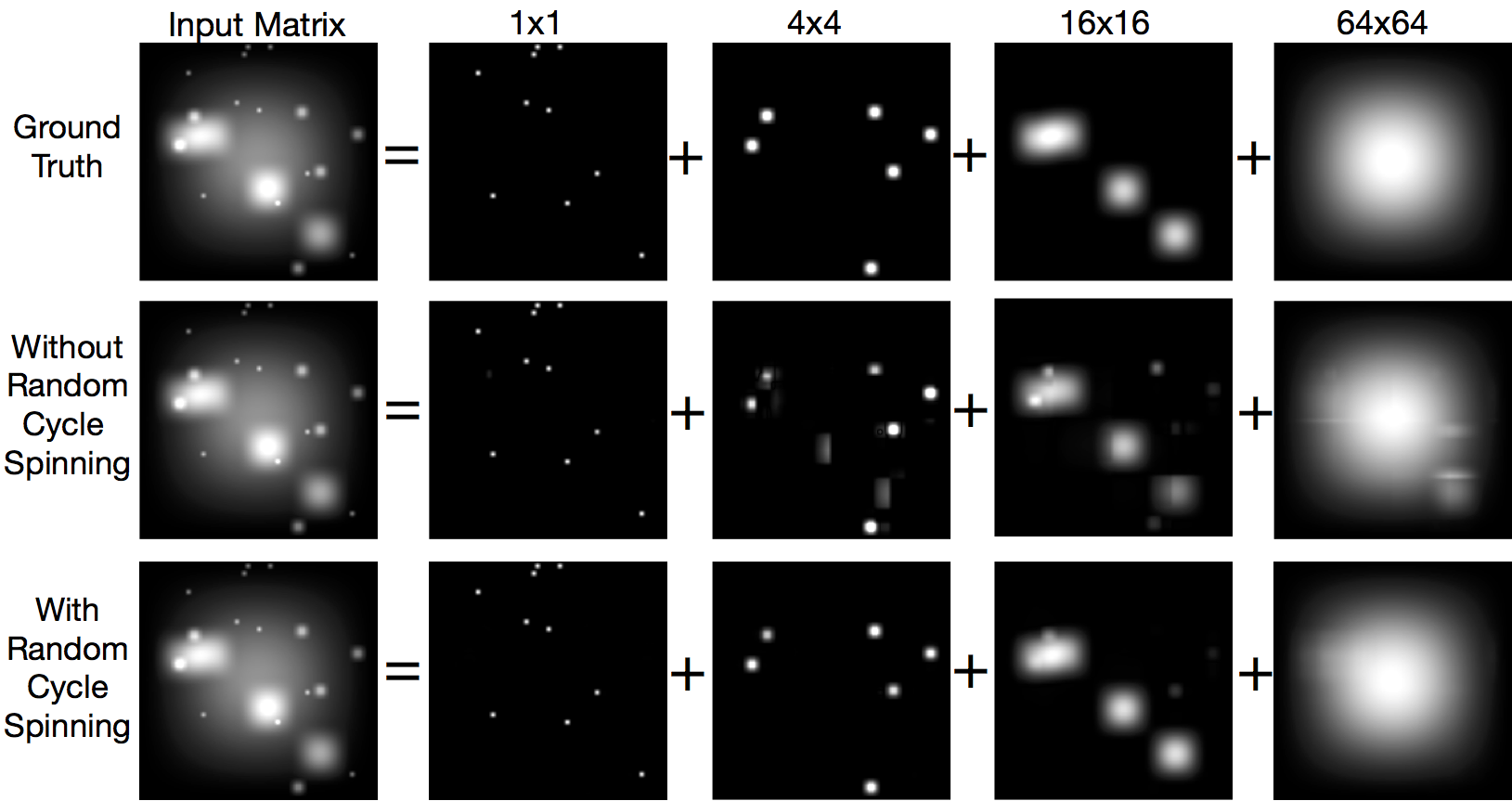}
\else
	\includegraphics[width=0.8\linewidth]{figures/hanning_shift.png}
\fi
\caption{An example of the multi-scale low rank decomposition with and without random cycle spinning. Each blob in the input matrix Y is a rank-1 matrix constructed from an outer product of hanning windows and is placed at random positions. Blocking artifacts can be seen in the decomposition without random cycle spinning while vastly diminished in the random cycle spinned decomposition.}
\label{fig:hanning_shift}
\end{center}
\end{figure}

Cycle spinning~\cite{Coifman:1995ji} has been commonly used in wavelet denoising to reduce the blocking artifacts due to the translation variant nature of the wavelet transform. To minimize artifacts, cycle spinning averages the denoised results from all possible \revised{shifted copies of the input}, thereby making the entire process translation invariant. \revised{\reviewercommentnum{R1.7} Concretely, let $S$ be the set of all shifts possible in the target application, $\textsc{Shift}_s$ denote the shifting operator by $s$, and $\textsc{Denoise}$ be the denoising operator of interest. Then the cycle spinned denoising of the input $X$ is given by:}
\begin{equation}
	\begin{aligned}
	\revised{\frac{1}{|S|} \sum_{s \in S} \textsc{Shift}_{-s} (\textsc{Denoise}  (\textsc{Shift}_s (X) ))}
	\end{aligned}
\end{equation}
In the context of multi-scale low rank decomposition, we can make the iterative algorithm translation invariant by replacing the block-wise singular value thresholding operation in each iteration with its cycle spinning counterpart. \revised{In particular, for our ADMM update steps, we can replace the $Z_i$ step to:}
\begin{equation}
	\begin{aligned}
	 \revised{Z_i \leftarrow \frac{1}{|S|} \sum_{s \in S} \textsc{Shift}_{-s} (\bsvt_{\lambda_i / \rho}  (\textsc{Shift}_s (X_i+U_i) ))}
	\end{aligned}
\end{equation}

To further reduce computation, we perform random cycle spinning in each iteration as described in Figueiredo et al.~\cite{Figueiredo:2003gd}, in which we randomly shifts the input, performs block-wise singular value thresholding and then unshifts back:
\begin{equation}
	\begin{aligned}
	 \revised{Z_i \leftarrow \textsc{Shift}_{-s} (\bsvt_{\lambda_i / \rho}  (\textsc{Shift}_s (X_i+U_i) ))}
	\end{aligned}
\end{equation}
\revised{where $s$ is randomly chosen from the set $S$.}

Using random cycle spinning, blocking artifacts caused by thresholding are averaged over iterations and in practice, reduces distortion significantly. Figure~\ref{fig:hanning_shift} shows an example of the multi-scale low rank decomposition with and without random cycle spinning applied on a simulated data that does not fall on the partition grid. The decomposition with random cycle spinning vastly reduces blocking artifacts that appeared in the one without random cycle spinning.



\section{Applications}
\label{sec:exp}

To test for practical performance, we applied the multi-scale low rank decomposition on four different real datasets that are conventionally used in low rank modeling: illumination normalization for face images (Section~\ref{ssec:face}), motion separation for surveillance videos (Section~\ref{ssec:surv}), multi-scale modeling of dynamic contrast enhanced magnetic resonance imaging (Section~\ref{ssec:dce}) and collaborative filtering exploiting age information (Section~\ref{ssec:collab}). We compared our proposed multi-scale low rank decomposition with low rank + sparse decomposition for the first three applications and with low rank matrix completion for the last application. Randomly cycle spinning was used for multi-scale low rank decomposition for all of our experiments. Regularization parameters $\lambda_i$ were chosen exactly as $\sqrt{ m_i }  + \sqrt{ n_i } + \sqrt{ \log (  M N  / \max  ( m_i, n_i )  )}$ for multi-scale low rank and $\max ( m_i, n_i )$ for low rank + sparse decomposition. Our simulations were written in the C programming language and ran on a 20-core Intel Xeon workstation. Some results are best viewed in video format, which are available as supplementary materials.

In the spirit of reproducible research, we provide a software package (in C and partially in MATLAB) to reproduce most of the results described in this paper. The software package can be downloaded from:
\begin{center}
\erased{\url{http://www.eecs.berkeley.edu/~mlustig/Software.html}}
\revised{ \url{https://github.com/frankong/multi_scale_low_rank.git} }
\end{center}

\subsection{Multi-scale Illumination Normalization for Face Recognition Pre-processing}
\label{ssec:face}

\begin{figure*}[!ht]
\begin{center}

\ifdefined \single
	\includegraphics[width=\linewidth]{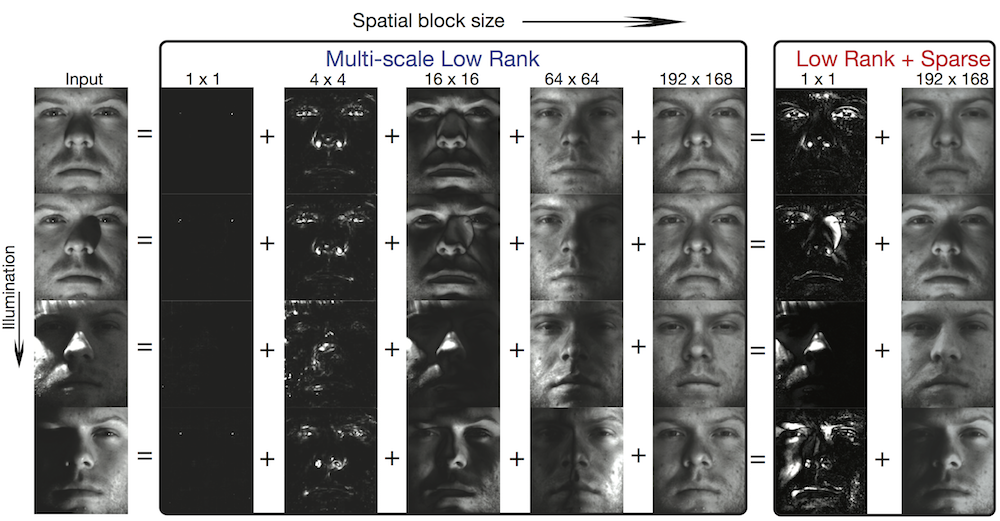}
\else
	\includegraphics[width=0.75\linewidth]{figures/face.png}
\fi

\caption{Multi-scale low rank versus low rank + sparse on faces with uneven illumination. Multi-scale low rank decomposition recovers almost shadow-free faces, whereas low rank + sparse decomposition can only remove some shadows.}
\label{fig:face}
\end{center}
\end{figure*}

Face recognition algorithms are sensitive to shadows or occlusions on faces. In order to obtain the best possible performance for these algorithms, it is desired to remove illumination variations and shadows on the face images. Low rank modeling are often used to model faces and is justified by approximating faces as convex Lambertian surfaces~\cite{Basri:2003ie}. 

Low rank + sparse decomposition~\cite{Candes:2011bd} was recently proposed to capture uneven illumination as sparse errors and was shown to remove small shadows while capturing the underlying faces as the low rank component. However, most shadows are not sparse and contain structure over different lighting conditions. Here, we propose modeling shadows and illumination changes in different face images as block-low rank as illumination variations are spatially correlated in multiple scales.

We considered face images from the Yale B face database~\cite{Georghiades:LXuok52i}. Each face image was of size $192 \times 168$ with $64$ different lighting conditions. The images were then reshaped into a $32,256 \times 64$ matrix and both multi-scale low rank and low rank + sparse decomposition were applied on the data matrix. For low rank + sparse decomposition, we found that the best separation result was obtained when each face image was normalized to the maximum value. For multi-scale low rank decomposition, the original unscaled image was used. Only the space dimension was decimated as we assumed there was no ordering in different illumination conditions. The multi-scale matrix partition can be visualized as in Figure~\ref{fig:blocks2}. 

Figure~\ref{fig:face} shows one of the comparison results. Multi-scale low rank decomposition recovered almost shadow-free faces. In particular, the sparkles in the eyes were represented in the $1 \times 1$ block size and the larger illumination changes were represented in bigger blocks, thus capturing most of the uneven illumination changes. In contrast, low rank + sparse decomposition could only recover from small illumination changes and still contained the larger shadows in the globally low rank component.

\subsection{Multi-scale Motion Separation for Surveillance Videos}
\label{ssec:surv}

\begin{figure*}[!ht]
\begin{center}
\ifdefined \single
\includegraphics[width=\linewidth]{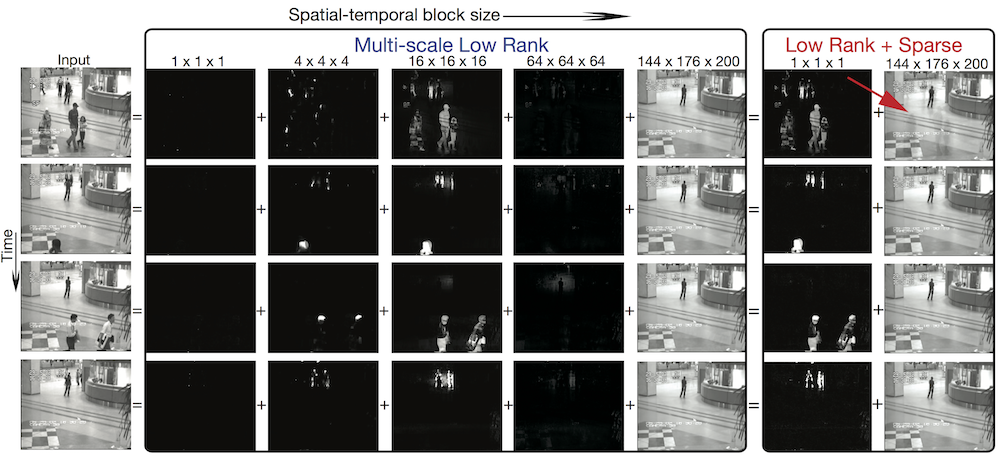}
\else
\includegraphics[width=0.8\linewidth]{figures/hall.png}
\fi

\caption{Multi-scale low rank versus low rank + sparse decomposition on a surveillance video. For the multi-scale low rank, body motion is mostly captured in the $16 \times 16 \times 16$ scale while fine-scale motion is captured in $4 \times 4 \times 4$ scale. Background video component is captured in the globally low rank component and is almost artifact-free. Low rank + sparse decomposition exhibits ghosting artifacts as pointed by the red arrow because they are neither globally low rank or sparse. }
\label{fig:hall}
\end{center}
\end{figure*}

In surveillance video processing, it is desired to extract foreground objects from the video. To be able to extract foreground objects, both the background and the foreground dynamics have to be modeled. Low rank modeling have been shown to be suitable for slowly varying videos, such as background illumination changes. In particular, if the video background only changes its brightness over time, then it can be represented as a rank-$1$ matrix. 

Low rank + sparse decomposition~\cite{Candes:2011bd} was proposed to foreground objects as sparse components and was shown to separate dynamics from background components. However, sparsity alone cannot capture motion compactly and often results in ghosting artifacts occurring around the foreground objects as shown in Figure~\ref{fig:hall}. Since video dynamics are correlated locally at multiple scales in space and time, we propose using the multi-scale low rank modeling with two sided decimation to capture different scales of video dynamics over space and time.

We considered a surveillance video from Li et al.~\cite{Li:2004fk}. Each video frame was of size $144 \times 176$ and the first $200$ frames were used. The video frames were then reshaped into a $25,344 \times 200$ matrix and both multi-scale low rank and low rank + sparse decomposition were applied on the data matrix.

Figure~\ref{fig:hall} shows one of the results. Multi-scale low rank decomposition recovered a mostly artifact free background video in the globally low rank component whereas low rank + sparse decomposition exhibits ghosting artifact in certain segments of the video. For the multi-scale low rank decomposition, body motion was mostly captured in the $16 \times 16 \times 16$ scale while fine-scale motion was captured in $4 \times 4 \times 4$ scale.

\subsection{Multi-scale Low Rank Modeling for Dynamic Contrast Enhanced Magnetic Resonance Imaging}
\label{ssec:dce}

\begin{figure*}[!ht]
\begin{center}

\ifdefined \single
\includegraphics[width=\linewidth]{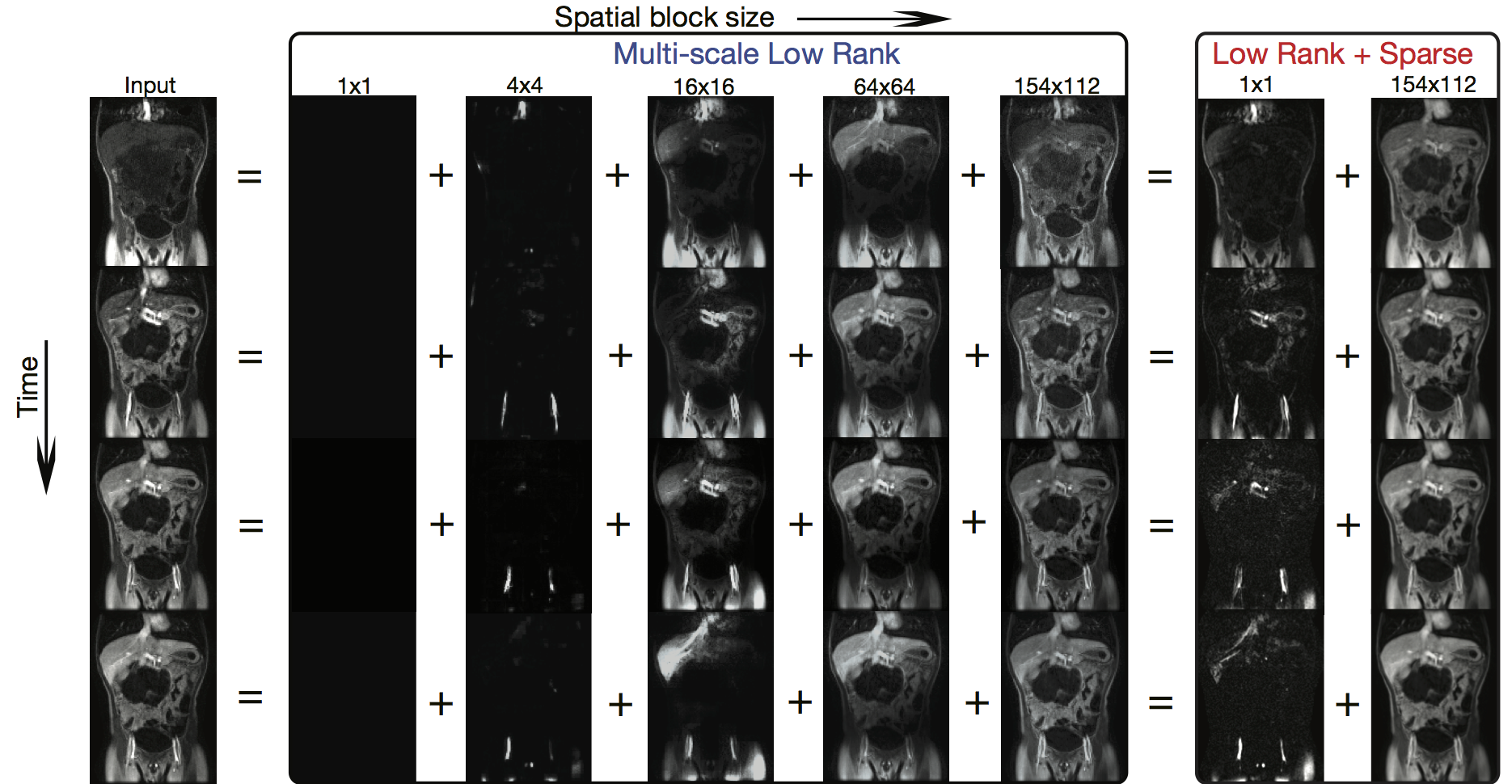}
\else
\includegraphics[width=0.75\linewidth]{figures/dce.png}
\fi

\caption{Multi-scale low rank versus low rank + sparse decomposition on a dynamic contrast enhanced magnetic resonance image series. For the multi-scale result, small contrast dynamics in vessels are captured in $4 \times 4$ blocks while contrast dynamics in the liver are captured in $16 \times16$ blocks. The biggest block size captures the static tissues and interestingly the respiratory motion. In contrast, the low rank + sparse modeling could only provide a coarse separation of dynamics and static tissue, which result in neither truly sparse nor truly low rank components.}
\label{fig:dce}
\end{center}
\end{figure*}

In dynamic contrast enhanced magnetic resonance imaging (DCE-MRI), a series of images over time is acquired after a $T_1$ contrast agent was injected into the patient. Different tissues then exhibit different contrast dynamics over time, thereby allowing radiologists to characterize and examine lesions. Compressed sensing Magnetic Resonance Imaging~\cite{Lustig:2007cu} is now a popular research approach used in three dimensional DCE-MRI to speed up acquisition. Since the more compact we can represent the image series, the better our compressed reconstruction result becomes, an accurate modeling of the dynamic image series is desired to improve the compressed sensing reconstruction results for DCE-MRI.

When a region contains only one type of tissue, then the block matrix constructed by stacking each frame as columns will have rank $1$. Hence, low rank modeling~\cite{Liang:2007gf}, and locally low rank modeling~\cite{Zhang:2015dva} have been popular models for DCE-MRI. Recently, low rank + sparse modeling~\cite{Otazo:2014it} have also been proposed to model the static background and dynamics as low rank and sparse matrices respectively. However, dynamics in DCE-MRI are almost never sparse and often exhibit correlation across different scales. Hence, we propose using a multi-scale low rank modeling to capture contrast dynamics over multiple scales.

We considered a fully sampled dynamic contrast enhanced image data. The data was acquired in a pediatric patient with 20 contrast phases, $1\times1.4\times2$ mm$^3$ resolution, and $~8$s temporal resolution. The acquisition was performed on a $3$T GE MR750 scanner with a 32-channel cardiac array using an RF-spoiled gradient-echo sequence. We considered a 2D slice of size $154 \times 112$ were then reshaped into a $17,248 \times 20$ matrix. Both multi-scale low rank and low rank + sparse decomposition were applied on the data matrix.

Figure~\ref{fig:dce} shows one of the results. In the multi-scale low rank decomposition result, small contrast dynamics in vessels were captured in $4 \times 4$ blocks while contrast dynamics in the liver were captured in $16 \times16$ blocks. The biggest block size captured the static tissues and interestingly the respiratory motion. Hence, different types of contrast dynamics were captured compactly in their suitable scales. In contrast, the low rank + sparse modeling could only provide a coarse separation of dynamics and static tissue, which resulted in neither truly sparse nor truly low rank components.

\subsection{Multi-scale Age Grouping for Collaborative Filtering}
\label{ssec:collab}

\begin{figure}[!ht]
\begin{center}
\includegraphics[width=\linewidth]{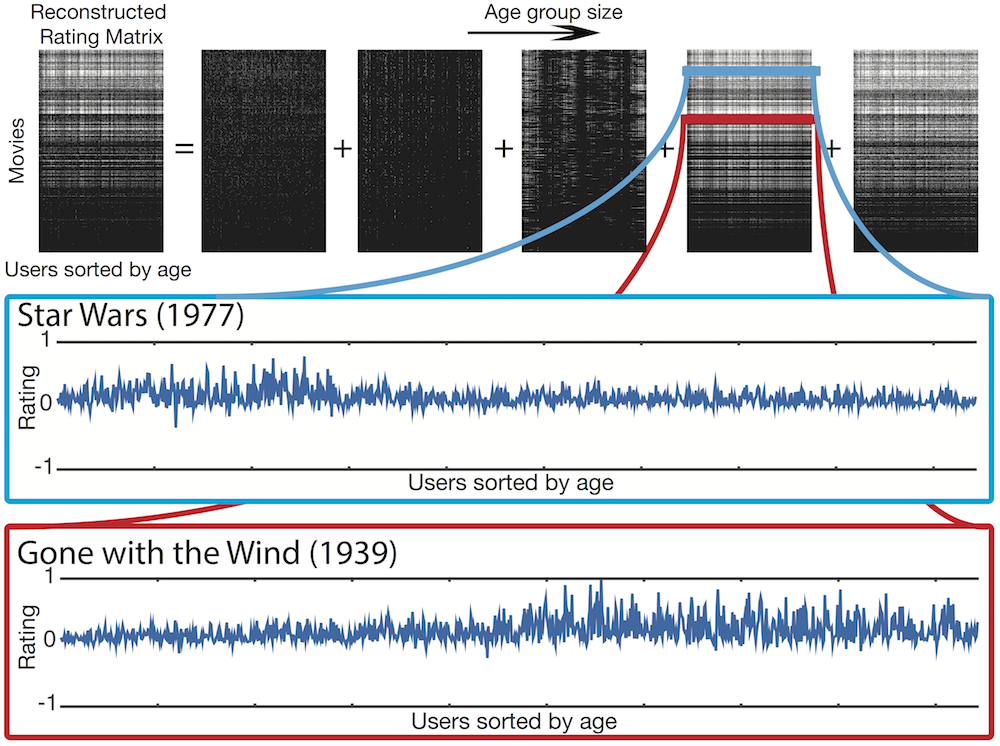}
\caption{Multi-scale low rank reconstructed matrix of the 100K MovieLens dataset. The extracted signal scale component captures the tendency that younger users rated Star Wars higher whereas the more senior users rated Gone with the Wind higher.}
\label{fig:collab}
\end{center}
\end{figure}

Collaborative filtering is the task of making predictions about the interests of a user using available information from all users. Since users often have similar taste for the same item, low rank modeling is commonly used to exploit the data similarity to complete the rating matrix~\cite{Candes:2010jb,Candes:2009kj,Recht:2010ht}. On the other hand, low rank matrix completion does not exploit the fact that users with similar demographic backgrounds have similar taste for similar items. In particular, users of similar age should have similar taste. Hence, we incorporated the proposed multi-scale low rank modeling with matrix completion by partitioning users according to their age and compared it with the conventional low rank matrix completion. \revised{Our method belongs to the general class of collaborative filtering methods that utilize demographic information~\cite{VOZALIS:2007ki}.}

To incorporate multi-scale low rank modeling into matrix completion, we change the data consistency constraint in problem~(\ref{eq:cvx}) to $[Y]_{jk} = [\sumi X_i]_{jk}$ for observed $jk$ entries, and correspondingly, the update step for $\seti{X_i}$ in equation~(\ref{eq:admm}) is changed to $[X_i]_{jk} \leftarrow [(Z_i - U_i) +  \frac{1}{L} ( Y - \sumi (Z_i - U_i) ) ]_{jk}$ for observed $jk$ entries and $[X_i]_{jk} \leftarrow [Z_i - U_i ]_{jk}$ for unobserved $jk$ entries. \revised{\reviewercommentnum{R1.2} We emphasize that our theoretical analysis does not cover matrix completion and the presented collaborative filtering application is mainly of empirical interest.}

To compare the methods, we considered the 100K MovieLens dataset, in which 943 users rated 1682 movies. The resulting matrix was of size $1682 \times 943$, where the first dimension represented movies and the second dimension represented users. The entire matrix had $93.7 \%$ missing entries. Test data was further generated by randomly undersampling the rating matrix by $5$. The algorithms were then run on the test data and root mean squared errors were calculated over all available entries.  To obtain a multi-scale partition of the matrix, we sorted the users according to their age along the second dimension and partitioned them evenly into age groups. 

Figure~\ref{fig:collab} shows a multi-scale low rank reconstructed user rating matrix. Using multiple scales of block-wise low rank matrices, correlations in different age groups were captured. For example, one of the scales shown in Figure~\ref{fig:collab} captures the tendency that younger users rated Star Wars higher whereas the more senior users rated Gone with the Wind higher. The multi-scale low rank reconstructed matrix achieved a root mean-squared-error of 0.9385 compared to a root mean-squared-error  of 0.9552 for the low rank reconstructed matrix.



\section{Discussion}
\label{discussion}

We have presented a multi-scale low rank matrix decomposition method that combines both multi-scale modeling and low rank matrix decomposition. Using a convex formulation, we can solve for the decomposition efficiently and exactly, provided that the multi-scale signal components are incoherent. We provided a theoretical analysis of the convex relaxation \revised{for exact decomposition}, which extends the analysis in Chandrasekaren et al.~\cite{Anonymous:2011kn}, \revised{and an analysis for approximate decomposition in the presence of additive noise, which extends the analysis in Agarwal et al.~\cite{Agarwal:2012gc}}. We also provided empirical results that the multi-scale low rank decomposition performs well on real datasets.

We would also like to emphasize that \revised{\reviewercommentnum{R1.1}\erased{the multi-scale low rank decomposition does not have any free parameter if the recommended regularization parameters are used, which}our recommended regularization parameters} empirically perform well even with the addition of noise, \revised{and hence in practice does not require manual tuning}. \revised{\erased{Hence, the multi-scale low rank prevents manual tuning and work ``right-out-of-the-box". }}\revised{While some form of }theoretical guarantees for the regularization parameters are \revised{provided in the approximate decomposition analysis, complete theoretical guarentees are not provided, especially for noiseless situations,} and would be valuable for future work. 

Our experiments show that the multi-scale low rank decomposition improves upon the low rank + sparse decomposition in a variety of applications. We believe that more improvement can be achieved if domain knowledge for each applications is incorporated with the multi-scale low rank decomposition. For example, for face shadow removal, prior knowledge of the illumination angle might be able to provide a better multi-scale partition. For movie rating collaborative filtering, general demographic information and movie types can be used to construct multi-scale partitions in addition to age information. 

\begin{appendices}

\section{Proof of Theorem~\ref{thm:main}}
\label{appendix:exact}

In this section, we provide a proof of Theorem~\ref{thm:main} and show that if $\seti{ X_i}$ satisfies a deterministic incoherence condition, then the proposed convex formulation (\ref{eq:cvx}) recovers $\seti{X_i}$ from $Y$ exactly. Our proof makes use of the dual certificate common in such proofs. \revisedtwo{We will begin by proving a technical lemma collecting three inequalities.}

\begin{lemma}
For $i = 1, \hdots, L$, the following three inequalities hold,
\begin{align*}
		 \dnormi{ \proj_{T_i}( X ) }{i}  &\le \dnormi{  X  }{i}  &\text{for any matrix } X \numberthis \label{eq:norm_ineq1} \\
		 \dnormi{ \proj_{T_i^\perp}( X ) }{i}  &\le \dnormi{  X  }{i} &\text{for any matrix } X \numberthis \label{eq:norm_ineq2} \\
		  \dnormi{ N_j } { i } &\le \mu_{ij}  \dnormi{ N_j } {j } &\text{for } j \ne i \text{ and }  N_j \in T_j	\numberthis	\label{eq:inc_prop}
\end{align*}
\label{lemma2}

\end{lemma}

\begin{proof}

To show the first inequality~\eqref{eq:norm_ineq1}, we recall that $\| X \|_{(i)}^* = \max_{b \in \part_i}  \msv{ R_{\block} ( X ) }$. Then, using the variational representation of the maximum singular value norm, we obtain,
\begin{align*}
  \dnormi{ \proj_{T_i}( X ) }{i} &=  \max_{b \in \part_i}~ \max_{u,v}~  u^\top  R_{\block} (\proj_{T_i} (X) ) v \\
  &= \max_{b \in \part_i}~ \max_{\substack{u \in \text{col}(R_b(X_i)) \text{ or } \\ v \in \text{row}(R_b(X_i))}} ~ u^\top R_{\block} (X) v\\
  &\le \max_{b \in \part_i} ~\max_{u,v}~ u^\top R_{\block} (X) v = \dnormi{ X }{i}
\end{align*}
where $\text{col}$ and $\text{row}$ denote the column and row spaces respectively.

Similarly, we obtain the second inequality~\eqref{eq:norm_ineq2}:
\begin{align*}
  \dnormi{ \proj_{T_i^\perp}( X ) }{i} &= \max_{b \in \part_i}~ \max_{\substack{u \in \text{col}^\perp(R_b(X_i)) \text{ and } \\ v \in \text{row}^\perp(R_b(X_i))}} ~ u^\top R_{\block} (X) v\\
  &\le \max_{b \in \part_i} ~\max_{u,v}~ u^\top R_{\block} (X) v = \dnormi{ X }{i}
\end{align*}

The third inequality~\eqref{eq:inc_prop} follows from the incoherence definition that $\mu_{ij} \ge \dnormi{N_j}{i} / \dnormi{N_j}{j}$ for any non-zero $N_j$.

\end{proof}

Next, we will show that if we can choose some parameters to ``balance" the coherence between the scales, then the block-wise row/column spaces $\seti{ T_i }$ are independent, that is $\sumi T_i$ is a direct sum. Consequently, each matrix $N$ in the span of $\seti{T_i}$ has a unique decomposition $N = \sumi N_i$, where $N_i \in T_i$.

\begin{prop}\label{prop:lin_ind}

If we can choose some positive parameters $\seti{ \lambda_i }$ such that
\begin{equation}
		\begin{aligned}
			\sum_{j \ne i}  \mu_{ij}   \frac{\lambda_j}{ \lambda_i} < 1, ~~~ \text{for } i = 1, \hdots, L
		\end{aligned}
\end{equation}
then we have
\begin{equation}
		\begin{aligned}
		T_i \cap \sum_{j \ne i} T_j = \{ 0 \}, ~~~ \text{for } i = 1, \hdots, L
		\end{aligned}
\end{equation}
In particular when $L = 2$, the condition on $\{ \mu_{12}, \mu_{21} \}$ reduces to $\mu_{12} \mu_{21} < 1$, which coincides with Proposition 1 in Chandrasekaren et al.~\cite{Anonymous:2011kn}. We also note that given $\mu_{ij}$, we can obtain $\seti{\lambda_i}$ that satisfies the condition $\sum_{j \ne i}  \mu_{ij}   \lambda_j <  \lambda_i$ by solving a linear program.

\end{prop}

\begin{proof}
\label{proof:lin_ind}

Suppose by contradiction that there exists $\seti{\lambda_i}$ such that $\sum_{j \ne i}  \mu_{ij}   \lambda_j /  \lambda_i < 1$, but $T_i \cap \sum_{j \ne i} T_j \ne \{ 0 \}$. Then there exists $\seti{ N_i \in T_i }$ such that $\sumi \lambda_i N_i = 0$ and not all $N_i$ are zero. But this leads to a contradiction because for $i = 1, \hdots, L$,
\begin{align*}
		\dnormi{ N_i }{i} &= \dnormi{ - \sum_{j \ne i} \frac{ \lambda_j }{ \lambda_i } N_j } {i}
		\\
		&\le \sum_{j \ne i} \frac{ \lambda_j }{ \lambda_i } \mu_{ij} \dnormi{N_j } {j}
		\numberthis \label{prop1_inc_prop}
		\\
		&\le (\sum_{j \ne i} \frac{ \lambda_j }{ \lambda_i } \mu_{ij})  \max_{j \ne i} \dnormi{N_j } {j}
		\numberthis \label{prop1_holder}
		\\
		&<   \max_{j \ne i} \dnormi{N_j } {j}
		\numberthis \label{prop1_assump}
\end{align*}
where we have used equation~(\ref{eq:inc_prop}) for the first inequality \eqref{prop1_inc_prop}, Holder's inequality for second inequality \eqref{prop1_holder} and $\sum_{j \ne i}  \mu_{ij}   \lambda_j /  \lambda_i < 1$ for the last inequality. \revised{Hence, none of $\seti{\dnormi{N_i}{i}}$ is the largest of the set, which is a contradiction.}

\end{proof}

Our next theorem shows an optimality condition of the convex program (\ref{eq:cvx}) in terms of its dual solution.

\begin{theorem} [Lemma 4.2~\cite{Anonymous:2013fx}] \label{thm:opt}

$\seti{ X_i } $  is the unique minimizer of the convex program (\ref{eq:cvx}) if there exists a matrix $\dual$ such that for $i = 1, \hdots, L$, 
\begin{enumerate}
	\item $ \proj_{T_i} ( \dual ) = \lambda_i \atom_i$
	\item $ \dnormi{ \proj_{T_i^\perp} (\dual ) } {i} < \lambda_i $
\end{enumerate}

\end{theorem}

\begin{proof}
 \label{proof:opt}

Consider any non-zero perturbation $\seti{\Delta_i}$ to $\seti{X_i}$ such that $\seti{X_i + \Delta_i}$ stays in the feasible set, that is $\sumi{\Delta_i} = 0$. We will show that $ \obj{ X_i + \Delta_i }  > \obj{ X_i }$.

We first decompose $\Delta_i$ into orthogonal parts with respect to $T_i$, that is, $\Delta_i =  \proj_{T_i}( \Delta_i ) + \proj_{T_i^\perp}( \Delta_i ) $. We also consider a specific subgradient $G = \transpose{  [ G_1 \cdots G_\levels ] }$ of $\sumi \lambda_i \normi{ \cdot }{i}$ at $\seti{ X_i }$ such that $\dnormi{ \proj_{T_i^\perp} ( G_i ) } {i} \le \lambda$,  and  $\inner{ \proj_{T_i^\perp}( \Delta_i ) }{ \proj_{T_i^\perp} ( G_i ) } = \lambda_i \normi{ \proj_{T_i^\perp}( \Delta_i ) }{i} $. \revised{\reviewercommentnum{R1.6} Then, from the definition of subgradient and the fact that $\sumi{\Delta_i} = 0$, we have,}
\begin{align*}
		 \obj{ X_i + \Delta_i }		&\ge \obj{ X_i } + \inner{ \Delta_i } { G_i } \\
		 					&= \obj{ X_i } + \inner{ \Delta_i } { G_i } - \inner{ \Delta_i } { \dual } 
\end{align*}
\revised{Applying the orthogonal decomposition with respect to $T_i$ and using $\proj_{T_i}(G_i) = \proj_{T_i}(Q) = \lambda_i E_i$, we have,}
\begin{align*}
		 \obj{ X_i + \Delta_i }		&\ge  \obj{ X_i } + \inner{ \proj_{T_i^\perp}( \Delta_i ) }{ \proj_{T_i^\perp} ( G_i ) }   \\
							      &~~~- \inner{ \proj_{T_i^\perp}( \Delta_i ) }{ \proj_{T_i^\perp} ( \dual ) }   \\
\end{align*}
\revised{Using Holder's inequality and the assumption for the subgradient $G_i$, we obtain,}
\begin{align*}
		 \obj{ X_i + \Delta_i }		&\ge  \obj{ X_i } +  \lambda_i \normi{ \proj_{T_i^\perp}( \Delta_i) }{i} \\
							&~~~  -    \dnormi{ \proj_{T_i^\perp}( \dual) }{i}  \normi{ \proj_{T_i^\perp}( \Delta_i) }{i} \\
							&> \obj{ X_i }
\end{align*}

\end{proof}

With Proposition~\ref{prop:lin_ind} and Theorem~\ref{thm:opt}, we are ready to prove Theorem~\ref{thm:main}.

\begin{proof} [Proof of Theorem~\ref{thm:main}]
\label{proof:main}

Since $\sum_{j \ne i}  \mu_{ij}   \lambda_j / \lambda_i < 1/2$, by Proposition~\ref{prop:lin_ind},  $T_i \cap \sum_{j \ne i} T_j = \{ 0 \}$ for all $i$. Thus, there is a unique matrix $Q$ in $\sumi T_i$ such that $\proj_{T_i} ( Q ) = \lambda_i \atom_i$. In addition, $Q$ can be uniquely expressed as a sum of elements in $T_i$. That is, $Q = \sumi Q_i$ with $Q_i \in T_i$. We now have a matrix $\dual$ that satisfies the first optimality condition. In the following, we will show that it also satisfies the second optimality condition $ \dnormi{ \proj_{T_i^\perp} { \dual } } {i} < \lambda_i $.

If the vector spaces $\seti{T_i}$ are orthogonal, then $Q_i$ is exactly $\lambda_i \atom_i$. Because they are not necessarily orthogonal, we express $Q_i$ as $ \lambda_i \atom_i$ plus a correction term $\lambda_i \epsilon_i$. That is, we express $Q_i = \lambda_i (\atom_i + \epsilon_i)$. Putting $Q_i$'s back to $Q$, we have
\begin{align*}
 		Q = \sumi \lambda_i (\atom_i + \epsilon_i)
	\numberthis \label{eq:q}
\end{align*}

\revised{\reviewercommentnum{R1.6} Combining the above equation \eqref{eq:q} with the first optimality condition \eqref{thm:opt}, $\proj_{T_i} ( Q ) = \lambda_i \atom_i $, we have $  \sum_{j=1}^L \lambda_j \proj_{T_i}(\atom_j + \epsilon_j) = \lambda_i \atom_i$. Since $\proj_{T_i} (\atom_i + \epsilon_i) = \atom_i + \epsilon_i$, rearranging the equation, we obtain the following recursive expression for $\epsilon_i$:}
\begin{equation}
	\begin{aligned}
		  \epsilon_i &= -  \proj_{T_i} \left ( \sum_{j \ne i} \frac{\lambda_j}{ \lambda_i}  (\atom_j  + \epsilon_j ) \right )
	\end{aligned}
	\label{eq:eps_recur}
\end{equation}
We now obtain a bound on $ \dnormi{  \proj_{T_i^\perp} ( Q ) }{i}$ in terms of $\epsilon_i$.
\begin{align*}
		\dnormi{  \proj_{T_i^\perp} ( Q )} {i}&= \dnormi{ \proj_{T_i^\perp} ( \sum_{j \ne i} \lambda_j  (\atom_j + \epsilon_j  )) }{i} 
		\\
		&\le \dnormi{  \sum_{j \ne i} \lambda_j  (\atom_j + \epsilon_j ) }{i} \numberthis \label{thm2_proj_bound}
		\\
		&\le   \sum_{j \ne i}  \mu_{ij}   \lambda_j ( 1 + \epsi{j})		\numberthis \label{thm2_inc_prop}
		\\
		&\le   (\sum_{j \ne i}  \mu_{ij}   \lambda_j  ) \max_{j \ne i} ( 1 + \epsi{j})	\numberthis \label{thm2_holder}
\end{align*}
where we obtain equation \eqref{thm2_proj_bound} from equation \eqref{eq:norm_ineq2}, equation \eqref{thm2_inc_prop} from equation \eqref{eq:inc_prop} and the last inequality \eqref{thm2_holder} from Holder's inequality.

Similarly, we obtain a recursive expression for $1 + \epsi{i}$ using equation~(\ref{eq:eps_recur})
\begin{align*}
		1 + \epsi{i}  &=  1 +  \dnormi{   \proj_{T_i} ( \sum_{j \ne i}  \frac{\lambda_j}{ \lambda_i}  (\atom_j  + \epsilon_j ) ) } {i}
		\\
		&\le 1 + \dnormi{  \sum_{j \ne i} \frac{\lambda_j}{ \lambda_i}   (\atom_j + \epsilon_j ) }{i} 	\numberthis \label{thm2_proj_bound2}
		\\
		&\le  1 +  \sum_{j \ne i}  \mu_{ij}   \frac{\lambda_j}{ \lambda_i}   ( 1 + \epsi{j})			\numberthis \label{thm2_inc_prop2}
		\\
		&\le  1 +  (\sum_{j \ne i}  \mu_{ij}   \frac{\lambda_j}{ \lambda_i} ) \max_{j \ne i}  ( 1 + \epsi{j})	\numberthis \label{thm2_holder2}
\end{align*}
where we obtain equation \eqref{thm2_proj_bound2} from equation \eqref{eq:norm_ineq1}, equation \eqref{thm2_inc_prop2} from equation \eqref{eq:inc_prop} and the last inequality \eqref{thm2_holder2} from Holder's inequality.

Taking the maximum over $i$ on both sides and rearranging, we have
\begin{align*}
		\max_i (1 + \epsi{i}) \le  \frac{ 1 } { 1 -   \max_i \sum_{j \ne i}  \mu_{ij}   \frac{\lambda_j}{ \lambda_i}  } 
\end{align*}
Putting the bound back to equation (\ref{thm2_holder}) , we obtain
\begin{equation}
	\begin{aligned}
		\dnormi{ \proj_{T_i^\perp} ( Q ) }{i}  &\le  \lambda_i  \frac{ \sum_{j \ne i}  \mu_{ij}   \frac{\lambda_j}{ \lambda_i}  } { 1 -   \max_i \sum_{j \ne i}  \mu_{ij}   \frac{\lambda_j}{ \lambda_i} }\\
		&<  \lambda_i  
		\end{aligned}
	\label{eq:done}
\end{equation}
where we used $\sum_{j \ne i}  \mu_{ij}   \lambda_j / \lambda_i < 1 / 2$ in the last inequality.

Thus, we have constructed a dual certificate $\dual$ that satisfies the optimality conditions (\ref{thm:opt}) and $\{ X_i \}_{i=1}^L$ is the unique optimizer of the convex problem (\ref{eq:cvx}).

\end{proof}

\section{Proof of Theorem~\ref{thm:noisy}}
\label{appendix:approximate}

In this section, we provide a proof of Theorem~\ref{thm:noisy}, showing that as long as we can choose our regularization parameters accordingly, we obtain a solution from the convex program~\eqref{eq:noisy} that is close to the ground truth $\seti{X_i}$.

We will begin by proving a technical lemma collecting three inequalities. Throughout the section, we will assume $X_\noise$ is non-zero for simplicity, so that the subgradient of $\fro{X_\noise}$ is exactly $X_\noise / \fro{X_\noise}$.

\begin{lemma}
For $i = 1, \hdots, L$, the following three inequalities hold,
\begin{align*}
\normi{X_i}{i} -  \normi{ X_i + \Delta_i }{i} &\le   \normi{ \proj_{T_i} (\Delta_i) }{i} - \normi{  \proj_{T_i^\perp} (\Delta_i) }{i}  \numberthis \label{lemma1.1} \\
\sumi \lambda_i \normi{  \proj_{T_i^\perp} (\Delta_i) }{i}  &\le  3 \sumi \lambda_i \normi{ \proj_{T_i} (\Delta_i) }{i}  \numberthis \label{lemma1.2} \\
\normi{ \proj_{T_i} (\Delta_i) }{i} &\le \sqrt{ 2 \sumb \rank_b} \fro{ \Delta_i } \numberthis \label{lemma1.3}
\end{align*}
\label{lemma1}

\end{lemma}

\begin{proof} We will prove the inequalities in order.

Let us choose  a subgradient $G_i = E_i + W_i$ of $\normi{X_i}{i}$ at $X_i$ such that $\inner{W_i }{ \proj_{T_i^\perp} (\Delta_i)} = \normi{ \proj_{T_i^\perp} ( \Delta_i ) }{i}$. Then, from the definition of the subgradient, we have,
\begin{align*}
	\normi{ X_i + \Delta_i }{i} &\ge \normi{X_i}{i} + \inner{ G_i} {\Delta_i }\\
	&=  \normi{X_i}{i} +  \inner{  E_i} { \proj_{T_i} (\Delta_i) } +  \normi{  \proj_{T_i^\perp} (\Delta_i) }{i} \\
	&\ge  \normi{X_i}{i} -  \normi{ \proj_{T_i} (\Delta_i) }{i} +  \normi{  \proj_{T_i^\perp} (\Delta_i) }{i} \numberthis \label{lemma_holder}
\end{align*}
where we used Holder's inequality for the last inequality \eqref{lemma_holder}. Re-arranging, we obtain the first result \eqref{lemma1.1}.

For the second inequality, we note that since $\sumi X_i + \Delta_i + X_\noise + \Delta_\noise = Y$, we have $\Delta_\noise = -\sumi \Delta_i$. From the definition of subgradient, we obtain,
\begin{align*}
\lambda_\noise \fro{ X_\noise + \Delta_\noise }  
 &\ge \lambda_\noise \fro{X_\noise } +  \lambda_\noise \inner{ \frac{X_\noise}{ \fro{X_\noise }}  }{\Delta_\noise}\\
 &= \lambda_\noise \fro{X_\noise } - \sumi \lambda_\noise \inner{ \frac{X_\noise}{ \fro{X_\noise }}  }{\Delta_i}\\
&\ge \lambda_\noise \fro{ X_\noise}  - \sumi \lambda_\noise \frac{\dnormi{X_\noise}{i}}{\fro{X_\noise } } \normi{\Delta_i}{i} \numberthis \label{lemma1.2.1}\\
&\ge \lambda_\noise \fro{ X_\noise }  - \sumi \frac{\lambda_i}{2}   \normi{\Delta_i}{i}  \numberthis \label{lemma1.2.2} \\
&\ge \lambda_\noise \fro{ X_\noise }  - \sumi \frac{\lambda_i}{2} \normi{ \proj_{T_i} (\Delta_i) }{i} \\
&- \frac{\lambda_i}{2}   \normi{  \proj_{T_i^\perp} (\Delta_i) }{i}  \numberthis \label{lemma1.2.3}
\end{align*}
where we obtain equation \eqref{lemma1.2.1} from Holder's inequality, equation \eqref{lemma1.2.2} from the condition of $\lambda_i$ \eqref{lambda_condition} and equation \eqref{lemma1.2.3} from the triangle inequality.

Since $\seti{X_i + \Delta_i}$ and $X_\noise + \Delta_\noise$ achieves the minimum objective function, we have,
\begin{align*}
& \lambda_\noise \fro{ X_\noise }  + \sumi \lambda_i \normi{ X_i }{i}\\
&\ge \lambda_\noise \fro{ X_\noise + \Delta_\noise }+ \sumi \lambda_i \normi{ X_i + \Delta_i }{i}   
\end{align*}
Substituting equation \eqref{lemma_holder} and \eqref{lemma1.2.3}, we obtain,
\begin{align*}
& \lambda_\noise \fro{ X_\noise }  + \sumi \lambda_i \normi{ X_i }{i}\\
&\ge \lambda_\noise \fro{ X_\noise }  - \sumi \frac{\lambda_i}{2} \normi{ \proj_{T_i} (\Delta_i) }{i} - \frac{\lambda_i}{2}   \normi{  \proj_{T_i^\perp} (\Delta_i) }{i}  \numberthis \\
&+  \sumi  \lambda_i \normi{X_i}{i} -  \lambda_i \normi{ \proj_{T_i} (\Delta_i) }{i} +  \lambda_i \normi{  \proj_{T_i^\perp} (\Delta_i) }{i} 
\end{align*}

Cancelling and re-arranging, we obtain the desired inequality \eqref{lemma1.2} ,
\begin{align*}
\sumi \lambda_i \normi{  \proj_{T_i^\perp} (\Delta_i) }{i}  &\le  3 \sumi \lambda_i \normi{ \proj_{T_i} (\Delta_i) }{i} 
\end{align*}

For the third inequality, recall that for any rank-$r$ matrix $X$, its nuclear norm $\nuc{X}$ is upper bounded by $\sqrt{r} \fro{X}$. Moreover, the projection of any matrix $Y$ to the column and row space $T$ of a rank $r$ matrix is at most rank-$2r$, that is $\text{rank}(\proj_T (Y)) \le 2r$. Hence, we obtain,
\begin{align*}
\normi{ \proj_{T_i} (\Delta_i) }{i}
&= \sumb \nuc{ R_b ( \proj_{T_i} (\Delta_i) ) } \\
&\le \sumb \sqrt{ 2 r_b } \fro{ R_b (\Delta_i ) }\\
&\le \sqrt{ \sumb 2 r_b } \fro{ \Delta_i }
\end{align*}
where the last inequality follows from Cauchy-Schwatz inequality and the fact that $\sumb \fro{ R_b (\Delta_i ) }^2 = \fro{\Delta_i}^2$.
\end{proof}

With these three inequalities, we now proceed to prove Theorem~\ref{thm:noisy}.

\begin{proof}[Proof of Theorem~\ref{thm:noisy}]

From the optimality of $\seti{ X_i + \Delta_i }$, we have the following inequality,
\begin{align*}
		&\lambda_\noise \fro{ X_\noise + \Delta_\noise }+ \sumi \lambda_i \normi{ X_i + \Delta_i }{i} \\
		&\le \lambda_\noise \fro{ X_\noise } + \sumi \lambda_i \normi{ X_i }{i}
\end{align*}
Re-arranging, we obtain,
\begin{align*}
		 &\fro{ X_\noise + \Delta_\noise } -  \fro{ X_\noise } \\
		 &\le  \frac{1}{\lambda_\noise} \sumi \lambda_i \left( \normi{ X_i }{i}-    \normi{ X_i + \Delta_i }{i} \right)
\end{align*}

For convenience, let us define $\Reg_T = \sumi \lambda_i \normi{ \proj_{T_i} (\Delta_i) }{i}$ and $\Reg_{T^\perp} = \sumi \lambda_i  \normi{  \proj_{T_i^\perp} (\Delta_i) }{i}$. Then, from Lemma \ref{lemma1} equation \eqref{lemma1.1}, we obtain,
\begin{align*}
		 \fro{ X_\noise + \Delta_\noise } -  \fro{ X_\noise } &\le  \frac{1}{\lambda_\noise} (\Reg_T - \Reg_{T^\perp}) \numberthis \label{thm3.1}
\end{align*}

We would like to keep only $\Delta_\noise$ on the left hand side and cancel $X_\noise$. To do this, we multiply both sides of equation \eqref{thm3.1} with $ \fro{ X_\noise + \Delta_\noise } +  \fro{ X_\noise }$. 
Then, using $(a+b)(a-b) = a^2 - b^2$, we expand the left hand side as:
\begin{align*}
	\text{L.H.S.}
		&= \fro{ X_\noise + \Delta_\noise }^2 -  \fro{ X_\noise }^2  \\
		&=	  \fro{ \Delta_\noise }^2  +  2 \inner{X_\noise}{  \Delta_\noise } 
\end{align*}
Recall that $\Delta_\noise = -\sumi \Delta_i$, we obtain the following lower bound for the left hand side:
\begin{align*}
	\text{L.H.S.}
		&=	  \fro{ \Delta_\noise }^2  -  2 \inner{X_\noise}{  \sumi \Delta_i }   \\
		&\ge	  \fro{ \Delta_\noise }^2  -  2  \sumi \dnormi{X_\noise}{i}  \normi{ \Delta_i }{i}  \numberthis \label{thm3.2}  \\
		&\ge	  \fro{ \Delta_\noise }^2  -   \frac{\fro{ X_\noise }}{\lambda_\noise} \sumi \lambda_i  \normi{ \Delta_i }{i}  \numberthis \label{thm3.3} \\
		&\ge	  \fro{ \Delta_\noise }^2  -   \frac{\fro{ X_\noise }}{\lambda_\noise} (\Reg_T + \Reg_{T^\perp})  \numberthis \label{thm3_lower}
\end{align*}
where we used Holder's inequality for equation \eqref{thm3.2}, the condition for $\lambda_i$ for equation \eqref{thm3.3}, and the triangle inequality for \eqref{thm3_lower}.

We now turn to upper bound the right hand side. We know $ \fro{ X_\noise + \Delta_\noise } \le \fro{ X_\noise }  + (\Reg_T - \Reg_{T^\perp}) / \lambda_\noise$ from equation \eqref{thm3.1}. Hence, we obtain,
\begin{align*}
		\text{R.H.S.}
		&= (\fro{ X_\noise + \Delta_\noise } +  \fro{ X_\noise }) \frac{1}{\lambda_\noise} (\Reg_T - \Reg_{T^\perp}) \\
		&\le ( 2 \fro{ X_\noise } + \frac{1}{\lambda_\noise} (\Reg_T - \Reg_{T^\perp})) \frac{1}{\lambda_\noise} (\Reg_T - \Reg_{T^\perp}) 
\end{align*}

Using Lemma \ref{lemma1} equation \eqref{lemma1.2}, we have,
\begin{align*}
		\text{R.H.S.}
		&\le 2   \frac{\fro{ X_\noise }}{\lambda_\noise} (\Reg_T - \Reg_{T^\perp}) + \frac{1 }{\lambda_\noise^2} (\Reg_T  - \Reg_{T^\perp})^2\\
		&\le 2   \frac{\fro{ X_\noise }}{\lambda_\noise} (\Reg_T - \Reg_{T^\perp}) + 16 \frac{1 }{\lambda_\noise^2} \Reg_T^2
		\numberthis \label{thm3_upper}
\end{align*}

Combining and simplifying the lower bound \eqref{thm3_lower} and the upper bound \eqref{thm3_upper}, we obtain,
\begin{align*}
		  \fro{ \Delta_\noise }^2  
		&\le 3  \frac{\fro{ X_\noise }}{\lambda_\noise}   \Reg_T + 16 \frac{1  }{\lambda_\noise^2} \Reg_T^2
		\numberthis \label{thm3_prelim_lower_bound}
\end{align*}

We will now lower bound $\fro{ \Delta_\noise }^2 = \fro{ \sumi \Delta_i }^2$ by individual terms:
\begin{align*}
		\fro{ \sumi \Delta_i }^2 
		&= \sumi \fro{  \Delta_i }^2  + \inner{\Delta_i }{  \sum_{j \ne i}  \Delta_j } \\
		&\ge \sumi \fro{  \Delta_i }^2  - \normi{ \Delta_i }{i}  \sum_{j \ne i}  \dnormi{ \Delta_j }{i} 
\end{align*}
where we used Holder's inequality for the last inequality.

Now, using the assumption that both $\dnormi{X_j}{i}$ and $\dnormi{X_j + \Delta_j}{i}$ are bounded by $\alpha_{ij}$. We have,
\begin{align*}
		\fro{ \sumi \Delta_i }^2 
		&\ge \sumi \fro{  \Delta_i }^2  - \normi{ \Delta_i }{i}  \sum_{j \ne i} 2 \alpha_{ij} \\
		&\ge \sumi \fro{  \Delta_i }^2  -  \lambda_i  \normi{ \Delta_i }{i}   \\
		&\ge \sumi \fro{  \Delta_i }^2  - \Reg_T - \Reg_{T^\perp}  \numberthis \label{thm3_triangle}\\
		&\ge \sumi \fro{  \Delta_i }^2  - 4\Reg_T  \numberthis \label{thm3_tt}
\end{align*}
where we used the triangle inequality for equation \eqref{thm3_triangle} and Lemma \ref{lemma1} equation \eqref{lemma1.2} for equation \eqref{thm3_tt}.

Substituting the lower bound back to equation \eqref{thm3_prelim_lower_bound}, we have
\begin{align*}
		  \sumi \fro{  \Delta_i }^2
		&\le (3 \frac{ \fro{ X_\noise}}{\lambda_\noise}  + 4) \Reg_T + 16 \frac{1  }{\lambda_\noise^2} \Reg_T^2 \numberthis \label{thm3_almost}
\end{align*}

We now turn to upper bound the equation. From Lemma \ref{lemma1} equation \eqref{lemma1.3}, we know that $\Reg_T \le \sumi \lambda_i \sqrt{ 2 \sumb \rank_b}  \fro{  \Delta_i }$. Hence, we have,
\begin{align*}
16 \Reg_T^2 &\le 16 \left( \sumi \lambda_i \sqrt{ 2 \sumb \rank_b}  \fro{  \Delta_i } \right)^2
\\
&\le 16 \left( \sumi \lambda_i^2 2 \sumb \rank_b \right)   \sumi \fro{  \Delta_i }^2
\numberthis \label{eq:1}
\\
&\le \frac{1}{2} \lambda_\noise^2 \sumi \fro{\Delta_i}^2
\numberthis \label{eq:2}
\end{align*}
where we used Cauchy-Schwartz's inequality for equation \eqref{eq:1} and the condition for $\lambda_\noise$ for equation \eqref{eq:2}.
Hence, substituting back to equation \eqref{thm3_almost}, rearranging and ignoring constants, we have,
\begin{align*}
		\sumi \fro{  \Delta_i }^2 
		&\lesssim \frac{ \fro{ X_\noise } }{ \lambda_\noise} \sumi \lambda_i \sqrt{ \sumb \rank_b}  \fro{  \Delta_i } 
\end{align*}

Completing the squares with respect to $\fro{\Delta_i}$ gives us,
\begin{align*}
		&\sumi \left( \fro{  \Delta_i } -  \frac{ \fro{ X_\noise } }{ \lambda_\noise}  \lambda_i \sqrt{ \sumb \rank_b} \right)^2 
		\\
		&\lesssim \sumi \left( \frac{ \fro{ X_\noise } }{ \lambda_\noise}  \lambda_i \sqrt{ \sumb \rank_b} \right)^2
\end{align*}

Using the triangle inequality to lower bound the left hand side, we obtain
\begin{align*}
		&\sumi  \fro{  \Delta_i } -  \frac{ \fro{ X_\noise } }{ \lambda_\noise}  \lambda_i \sqrt{ \sumb \rank_b} \\
		&\lesssim  \sqrt{\sumi \left( \frac{ \fro{ X_\noise } }{ \lambda_\noise}  \lambda_i \sqrt{ \sumb \rank_b} \right)^2}
\end{align*}

Using the fact that $\ell 1$-norm is larger than the $\ell 2$-norm, and re-arranging give us the desired result,
\begin{align*}
		  \sumi \fro{  \Delta_i }
		&\lesssim  \frac{\fro{ X_\noise }}{\lambda_\noise}  \sumi \lambda_i \sqrt{\sumb \rank_b }
\end{align*}

\end{proof}

\end{appendices}


\endgroup

\bibliography{paper}

\begin{IEEEbiography}[{\includegraphics[width=1in,height=1.25in,clip,keepaspectratio]{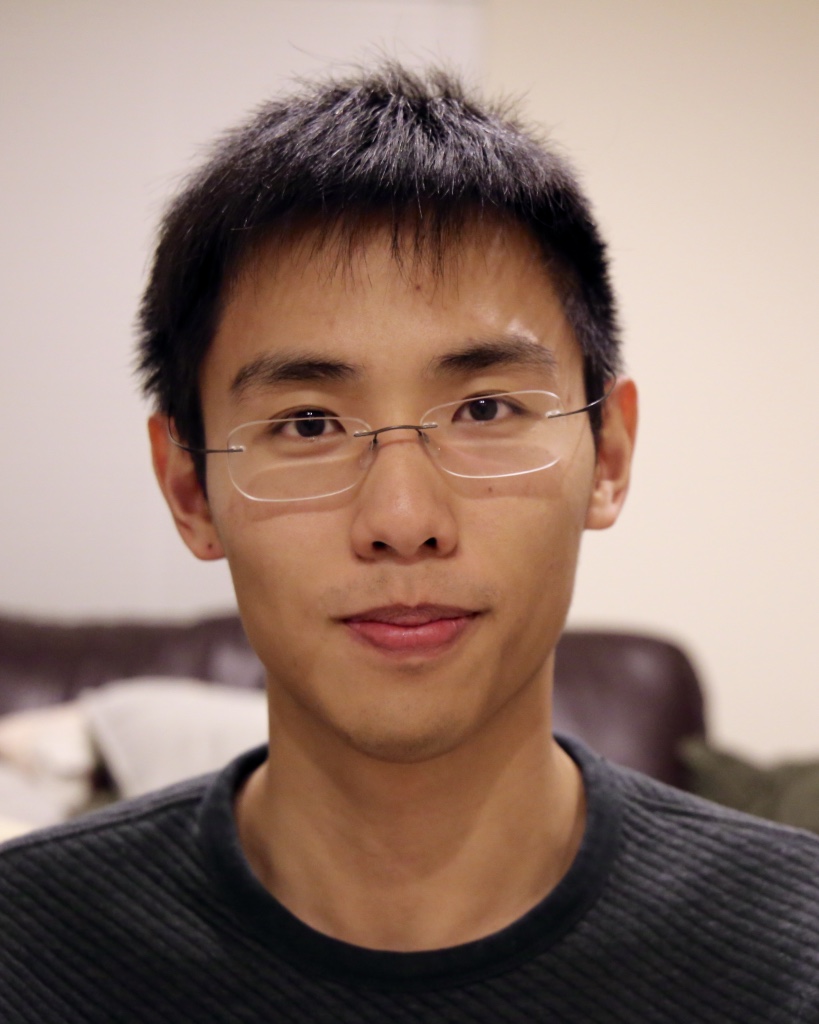}}]{Frank Ong} is currently pursuing the Ph.D. degree in the department of Electrical Engineering and Computer Sciences at the University of California, Berkeley. He received his B.Sc. at the same department in 2013. His research interests lie in signal processing, medical image reconstruction, compressed sensing and low rank methods.
\end{IEEEbiography}

\begin{IEEEbiography}[{\includegraphics[width=1in,height=1.25in,clip,keepaspectratio]{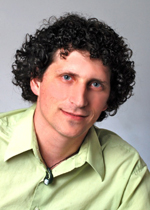}}]{Michael (Miki) Lustig} is an Associate Professor in EECS. He joined the faculty of the EECS Department at UC Berkeley in Spring 2010. He received his B.Sc. in Electrical Engineering from the Technion, Israel Institute of Technology in 2002. He received his Msc and Ph.D. in Electrical Engineering from Stanford University in 2004 and 2008, respectively. His research focuses on medical imaging, particularly Magnetic Resonance Imaging (MRI), and very specifically, the application of compressed sensing to rapid and high-resolution MRI, MRI pulse sequence design, medical image reconstruction, inverse problems in medical imaging and sparse signal representation.
\end{IEEEbiography}

\end{document}